\documentclass[12pt, draftclsnofoot, onecolumn]{IEEEtran}

\usepackage{amsmath, amssymb, amsthm,cite}
\usepackage[font=small]{caption}
\usepackage{graphicx,color}
\usepackage[update,prepend]{epstopdf}		
\usepackage{tabulary}
\usepackage{enumerate}
\usepackage{setspace}	
\title{Optimizing Data Aggregation for Uplink Machine-to-Machine Communication Networks}
\author{Derya~Malak,~Harpreet S. Dhillon, and~Jeffrey~G.~Andrews
\thanks{D. Malak and J. G. Andrews are with the Wireless and Networking Communications Group (WNCG), The University of Texas at Austin, TX, USA. Email: deryamalak@utexas.edu, jandrews@ece.utexas.edu. H. S. Dhillon is with Wireless@VT, Department of Electrical and Computer Engineering, Virginia Tech, Blacksburg, VA, USA. Email: hdhillon@vt.edu.}
\thanks{Manuscript last revised: {\today}.}}

\newcommand{\norm}[1]{\left\| #1\right\|}
                 
\newtheorem{theo}{Theorem}
\newtheorem{defi}{Definition}
\newtheorem{assu}{Assumption}
\newtheorem{prop}{Proposition}
\newtheorem{remark}{Remark}

\newtheorem{cor}{Corollary}

\newtheorem{lem}{Lemma}

\DeclareMathOperator*{\argmin}{arg\,min}

\begin{document}

\maketitle
\begin{abstract}
Machine-to-machine (M2M) communication's severe power limitations challenge the interconnectivity, access management, and reliable communication of data. In densely deployed M2M networks, controlling and aggregating the generated data is critical. We propose an energy efficient data aggregation scheme for a hierarchical M2M network. We develop a coverage probability-based optimal data aggregation scheme for M2M devices to minimize the average total energy expenditure per unit area per unit time or simply the {\em energy density} of an M2M communication network. Our analysis exposes the key tradeoffs between the energy density of the M2M network and the coverage characteristics for successive and parallel transmission schemes that can be either half-duplex or full-duplex. Comparing the rate and energy performances of the transmission models, we observe that successive mode and half-duplex parallel mode have better coverage characteristics compared to full-duplex parallel scheme. Simulation results show that the uplink coverage characteristics dominate the trend of the energy consumption for both successive and parallel schemes.
\end{abstract}

%\begin{IEEEkeywords}
%Machine-to-machine communication, aggregator, energy efficiency, rate, coverage, network hierarchy, stochastic geometry.
%\end{IEEEkeywords}

\maketitle

%%%%%%%%%%%%%%%%%%%%%%%%%%%%%%%%%%%%%%%%%%%%%%%%%%%%%%%%%%%%%%%%%%%%%%%%%%%%%%%%%%%%%%%%%%%%%%%%%%%%%%%%%%%%%%%%%%%%%%%%%%
\section{Introduction}
\label{intro}
Machine-to-machine (M2M) applications are rapidly growing, and will become an increasingly important source of traffic and revenue on 4G and 5G cellular networks. Unlike video applications, which are expected to consume around 70\% of all wireless data by the end of the decade \cite{Cisco2015}, M2M devices will use a comparatively small fraction. However, M2M communication has its own challenges. The air interface design for high-data-rate applications may not effectively support M2M's vast number of devices, each usually having only a small amount of data to transmit. Thus, M2M will require sophisticated access management and resource allocation with QoS constraints to prevent debilitating random access channel (RACH) congestion \cite{DhiHuVis2015}, and to enable link adaptation with low overhead, reduced energy consumption, and efficient control channel design \cite{Lien2011}. Bluetooth (IEEE 802.15.1), Zigbee (IEEE 802.15.4) and WiFi (IEEE 802.11b) are a few examples of current technologies that have been used for M2M communication \cite{Lien2011}. Meanwhile, 3GPP has been studying M2M communication in its standardization process for Long Term Evolution-Advanced (LTE-A). In Release 13, a low-complexity user equipment (UE) category for M2M devices is specified, in which a UE has reduced bandwidth and lower maximum transmission power, and can operate in all the duplex modes \cite{3GPPStandardv13}. However, there is still no consensus on the general network architecture for large scale M2M communication. 

\subsection{Motivation and Related Work}  
Unlike most human generated or consumed traffic, M2M as defined in this paper is characterized by a very large number of small transactions, often from battery powered devices. The power and energy optimal uplink design for various access strategies is studied in \cite{DhillonEH2-2013}, while an optimal uncoordinated strategy to maximize the average throughput for a time-slotted RACH is developed in \cite{DhillonEH2014}. For the small payload sizes relevant for M2M, a strategy that transmits both identity and data over the RACH is shown to support significantly more devices compared to the conventional approach, where transmissions are scheduled after an initial random-access stage. 

Clustering is a key technique to reduce energy consumption and it increases the scalability and lifetime of the network \cite{Ester1996,Bandyopadhyay2003}. A Low-Energy Adaptive Clustering Hierarchy protocol (LEACH) to evenly distribute the energy load among the sensors and enable scalability for dynamic networks by incorporating localized coordination is proposed \cite{Heinzelman2000energy,Kaj2009}. A low overhead protocol that periodically selects cluster-heads according to the node residual energy to support scalable data aggregation is established in \cite{Younis2004}. A Distributed Energy-Efficient Clustering (DEEC) scheme, where the cluster-heads are elected by a probability based on the ratio between residual energy of each node and the average energy of the network is proposed in \cite{Qing2006}. A data gathering mechanism, where clusters closer to the BS have smaller sizes than those farther away from the BS to balance the energy consumption is analyzed in \cite{Li2005MobileAdhoc}. Although these techniques are energy efficient, they do not incorporate the coverage characteristics of wireless transmissions.

A stochastic analysis to determine how many data collectors per area are required to meet the outage probability for a given sensor node density is implemented in \cite{Kwon2013}. Modulation strategies to minimize the total energy to send a given number of bits are analyzed in \cite{Cui2003}. The energy efficiency of a large scale interference limited ad hoc network at physical, medium access control (MAC) and routing layers is quantified \cite{Zaidi2012}. Optimal power allocation to minimize the total energy consumption in a clustered WSN where sensors cooperatively relay data to nearby clusters is studied in \cite{Zhou2008}. The energy-optimal relay distance and the optimal energy to transmit a bit successfully over a particular distance are analyzed in \cite{Holland2011}.  

Different control mechanisms to avoid congestion caused by random channel access of M2M devices are reviewed \cite{Hasan2013}. An adaptive slotted ALOHA scheme for random access control of M2M systems with bursty traffic that achieves near-optimal access delay performance is proposed \cite{wu2013fasa}. A comprehensive performance evaluation of the energy efficiency of the random access mechanism of LTE and its role for M2M is provided \cite{Laya2014}. An energy-efficient uplink design for LTE networks in M2M and human-to-human coexistence scenarios that satisfies quality-of-service (QoS) requirements is developed \cite{Aijaz2014}. Similar to \cite{Laya2014,Aijaz2014}, we study an energy-efficient design for M2M uplink where devices perform multi-hop transmissions. We also incorporate the coverage characteristics for different transmission modes using stochastic geometry.

Because low-power M2M devices may not be able to communicate with the BS directly, hierarchical architectures may be necessary. Hence, critical design issues also include optimizing hierarchical organization of the devices and energy efficient data aggregation. Although these issues have not been studied in the context of M2M, there is prior work on distributed networks in the context of wired communications. In \cite{Rodoplu1999}, energy consumption is optimized by studying a distributed protocol for stationary ad hoc networks. In \cite{Baccelli1999}, a distribution problem which consists of subscribers, distribution and concentration points for a wired network model is studied to minimize the cost by optimizing the density of distribution points. In \cite{Baek2004}, a hierarchical network including sinks, aggregators and sensors is proposed, which yields significant energy savings.

Hierarchical networks can provide efficient data aggregation in M2M or other power-limited systems to enable successful end-to-end transmission. Despite previous research efforts, e.g., \cite{Kwon2013}, \cite{Bandyopadhyay2003} and \cite{Holland2011}, to the best of our knowledge, there has been no study focusing on data aggregation schemes for M2M networks together with the rate coverage characteristics, especially from an energy optimal design perspective. Providing such a study is the main contribution of this paper.

\subsection{Contributions and Organization}
{\bf A large-scale hierarchical wireless network model for M2M.} We propose a new communication model for the M2M uplink. In Sect. \ref{systemmodel}, we consider a hierarchical scenario to model the wireless transmissions where the hierarchical levels describe the multi-hop stages of the model. Each level is composed of the transmitter and aggregator processes. Once the transmissions of the current hierarchical level are completed, the transmitters are either turned off or kept on depending on the transmission scheme used, and in the subsequent hierarchical level, the aggregator process forms the new transmitter process. Since each device can act as a transmitter or an aggregator in a particular time slot, we model the devices as wireless transceivers. In Sect. \ref{multi-stage}, the aggregation process is repeated over multiple levels to generate a hierarchical model. We show there is an optimal fraction of aggregators that minimizes the overall energy consumption.

{\bf SIR and rate coverage probability.} In Sect. \ref{coverageprobability}, we analyze the SIR coverage and rate coverage characteristics of the multi-stage transmission process to determine the optimal number of stages. We consider two possible transmission techniques: i) {\em successive scheme}, where the hierarchical levels are not active simultaneously, and ii) {\em parallel scheme}, where either all the levels are active simultaneously as in full-duplex transmission or where the active levels are interleaved as in half-duplex mode. We describe the proposed transmission models in detail in Sect. \ref{rate}. 

{\bf Optimizing the number of multi-hop stages.} We propose a general aggregator model for power limited M2M devices. Since the M2M devices are power limited and thus range limited, multi-hop routing is a feasible strategy rather than direct transmissions. However, in designing multi-hop protocols, the number of hops cannot be increased arbitrarily due to the additional energy consumption incurred by relays; long-hop routing is a competitive strategy for many networks \cite{Haenggi2005}. We find an upper and lower bound on the optimal number of hops in Sect. \ref{coverageprobability}.

{\bf Optimizing the network energy density.} We use a generic transceiver scheme to model the energy consumption of the devices. For multi-hop transmission models, we find the optimal fraction of aggregators independently chosen from the set of devices that minimizes the average total energy consumption per unit area, i.e., the mean total energy density, for a fixed payload per device of an M2M network. We incorporate the energy consumption of aggregators and transceiver circuit components into the energy model. In Sects. \ref{single-stage} and \ref{multi-stage}, we optimize the total energy density of the transmissions based on a given SIR coverage requirement for the devices.

We evaluate the performance of the proposed techniques in Sect. \ref{Performance}, and provide a comparison in terms of their energy densities and communication rates through numerical investigations. 

%%%%%%%%%%%%%%%%%%%%%%%%%%%%%%%%%%%%%%%%%%%%%%%%%%%%%%%%%%%%%%%%%%%%%%%%%%%%%%%%%%%%%%%%%%%%%%%%%%%%%%%%%%%%%%%%%%%%%%%%%%
\section{System Model}
\label{systemmodel}
We consider a cellular-based uplink model for M2M communication where the BS and device locations are distributed as independent Poisson Point Processes (PPPs)\footnote{PPP is not just plausible, basically a tractable model where the points are randomly and independently scattered in the space, and analysis for a typical node is permissible in a homogeneous PPP by Slivnyak's theorem \cite{Stoyan1996}.} with respective densities of $\lambda_{\mathrm{BS}}$ and $\lambda$ with $\lambda\gg \lambda_{\mathrm{BS}}$. Each BS has an average coverage area of $\lambda_{\mathrm{BS}}^{-1}$, and each device has a fixed payload of $M$ bits to be transmitted to the BS. We also assume open loop power control with maximum transmit power constraint under which the transmit power of a device located at distance $d$ from the BS is\footnote{Later in Sect. \ref{coverageprobability}, in evaluating the SIR-based coverage probability, we also incorporate the small-scale fading into the analysis that is assumed to be independent and identically distributed (iid) with unit mean. Therefore, incorporating fading yields the same average energy analysis. To keep the notation simple, we do not incorporate fading in Sects. \ref{systemmodel}, \ref{single-stage} and \ref{multi-stage}.} $P_T(d)=\min\{P_{T_{\max}}, \overline{P}_T d^{\alpha}\}$, 
where $\alpha$ is the path loss exponent, $P_{T_{\max}}$ is the maximum transmit power constraint and $\overline{P}_T$ is the received power when $d\leq \big({P_{T_{\max}}}/{\overline{P}_T}\big)^{1/\alpha}$. With this assumption, the average received power at the BS from a device in its coverage area and located at distance $d$ from the BS is constant and equal to $P_R(d)=\min\{P_{T_{\max}}d^{-\alpha}, \overline{P}_T \}$. Assuming $\mathrm{N_a}$ devices\footnote{$\mathrm{N_a}$ (random variable) denotes the number of devices in the Voronoi cell of a typical aggregator, and is detailed in Sect. \ref{single-stage}.} are scheduled based on TDMA, the uplink SINR for any device is ${\rm SINR} = {P_R(d)}{(I_{\rm ic} + I_{\rm oc} +N_0 W)^{-1}}$, where $I_{\rm ic}$ and $I_{\rm oc}$ are the intra cell and out of cell interferences\footnote{The interference is due to simultaneously active aggregator cells. Users within each Voronoi cell are assumed to use TDMA for access, and at a particular time slot, there is only one active transmitting device in each cell. For the sequential mode, the interference is due to the active transmitters outside the typical cell. On the other hand, for the parallel transmission mode, since all the stages are simultaneously active, there is both intra cell and out of cell interferences., which is detailed in Sect. \ref{rate}.}, $N_0$ is the power spectral density and $W$ is the bandwidth.

%%%%%%%%%%%%%%%%%%%%%%%%%%%%%%%%%%%%%%%%%%%%%%%%%%%%%%%%%%%%%%%%%%%%%%%%%%%%%%%%%%%%%%%%%%%%%%%%%%%%%%%%%%%%%%%%%%%%%%%%%%
\subsection{Data Aggregation and Transmission Model}
\label{data-aggregation}
Devices transmit data to the BS by aggregating data. The initial device process $\Psi$ is independently thinned by probability $\gamma<0.5$ to generate\footnote{In Sect. \ref{single-stage}, we will motivate the choice of $\gamma<0.5$ in our setup.} the aggregator process $\Psi_a$ with density $\lambda_a$ and the device (transmitter) process $\Psi_u$ with density $\lambda_u$, where $\lambda=\lambda_u+\lambda_a$. Each transmitter is associated to the closest receiving device (aggregator), i.e., the transmitting devices within the Voronoi cell of the typical aggregator device will transmit their payloads to that aggregator.  

The aggregation process can be extended to multiple stages. Each hierarchical level is composed of the transmitter and the aggregator processes, where the aggregator processes of all stages are initially determined such that they are disjoint from each other. At each stage, after the set of transmitters transmit their payloads to their nearest aggregators, and once the transmissions of a hierarchical level are completed, the transmitters are turned off and excluded from the process. In the subsequent hierarchical level, the aggregators of the previous stage become the transmitters, and they transmit their data to the aggregator process of the new stage. The aggregation process is repeated over multiple stages to generate the hierarchical transmission model. The process ends when all the payload is transmitted to the BS in the last stage of this multi-hop process, which we call a transmission cycle. The aggregation model will be detailed in Sects. \ref{single-stage} and \ref{multi-stage}.

%%%%%%%%%%%%%%%%%%%%%%%%%%%%%%%%%%%%%%%%%%%%%%%%%%%%%%%%%%%%%%%%%%%%%%%%%%%%%%%%%%%%%%%%%%%%%%%%%%%%%%%%%%%%%%%%%%%%%%%%%%
\subsection{Wireless Transceiver Model}
We model the devices as wireless transceivers, since each device can act as a transmitter or an aggregator in a particular time slot. In \cite{Wang2006}, a generic architecture for energy-limited wireless transceivers is provided. The transceiver has four major building blocks. The transmit block (TX) is responsible for modulation and up-conversion, the receive block (RX) for down-conversion and demodulation, the local oscillator (LO) block for the generation of the required carrier frequency, and the power amplifier (PA) block for amplification of the signal to produce the required RF transmit power $P_T$, where we assume the maximum transmit power is bounded and given by $P_{T_{\max}}$. The power consumption in the receive and transmit paths are
\begin{align}
P_{\rm rxr}=P_{\rm LO}+P_{\rm RX}+P_{\rm O},\quad P_{\rm txr}=P_{\rm LO}+P_{\rm TX}+P_{\rm PA},
\end{align}
where the power consumption of ${\rm LO}$, ${\rm RX}$ and ${\rm TX}$ blocks are denoted by $P_{\rm LO}$, $P_{\rm RX}$ and $P_{\rm TX}$, respectively, which are non negative constants. $P_{\mathrm{O}}$ is the receiver overhead power that is assumed to be constant. The PA power consumption is given by $P_{\rm PA}=\eta^{-1}P_T$, where $\eta$ is the PA efficiency, which is constant in the linear regime. The average energy cost of the transceiver is given by $\beta_R\cdot P_{\rm rxr}+\beta_T\cdot P_{\rm txr}$ where $\beta_T={\bar{t}_{\rm tx}}/{(\bar{t}_{\rm tx}+\bar{t}_{\rm rx})}$ and $\beta_R={\bar{t}_{\rm rx}}/{(\bar{t}_{\rm tx}+\bar{t}_{\rm rx})}$, and $\bar{t}_{\rm rx}$ and $\bar{t}_{\rm tx}$ stand for the average receive and transmit times. The symbol definitions are given in Table \ref{table:tab1}.

In the proposed model, multiple devices (transmitters) send data to an aggregator (receiver), where each device is allocated a different time slot on the same frequency, i.e., TDMA. Therefore, $\beta_T=1$ and $\beta_R=0$ for transmitting devices, and their energy cost is proportional to $P_{\rm txr}$. For the aggregator, $\beta_T=0$ and $\beta_R=1$, and its energy cost is proportional to $P_{\rm rxr}$. Let $E_R$ denote the average energy required by an aggregator. $E_T$ denotes the average energy required for all transmissions within the coverage of an aggregator node, i.e., if $\mathrm{N_a}$ devices transmit to the aggregator, $E_T$ is the sum of their average transmission energies. Hence, for a total transmission/reception time slot of duration $\Delta t$, aggregator and total transmitter energy consumptions are 
\begin{align}
\label{ER}
E_R=\mathbb{E}[\Delta t (P_{\rm LO}+P_{\rm RX}+P_{\rm{O}})],\quad E_T=\mathbb{E}[\Delta t \left(\mathrm{N_a}P_{\rm LO}+P_{\rm TX}+P_{\rm PA}\right)],
\end{align}
where we assume $\mathrm{N_a}$ devices sequentially transmit to the aggregator, the time slot satisfies $\Delta t=\mathbb{E}[\mathrm{N_a}]\bar{t}_{\rm tx}=\bar{t}_{\rm rx}\geq\bar{t}_{\rm tx}$, and we assume that the communication delay due to processing of the data is negligible. Hence, transmitted data can be received within the time slot allocated and the aggregator can decode the received data. We also assume that the transmissions and receptions are synchronized. $E_R$ is mainly determined by the reception duration, while $E_T$ depends on the energy consumption of the PA, $\bar{t}_{\rm tx}$ and $\mathrm{N_a}$. Furthermore, whenever a device is in sleep mode, its transmitter and receiver modules are not active but other components, i.e., ${\rm LO}$ and ${\rm RX}$ blocks, are still consuming energy that justifies the scaled term $\mathrm{N_a}P_{\rm LO}$ of $E_T$ in (\ref{ER}).

We will use this transceiver model and the data aggregation strategy described to formulate our energy optimization problems in Sects. \ref{single-stage}-\ref{multi-stage}, and analyze the SIR coverage in Sect. \ref{coverageprobability}.

\begin{table}[t!]\scriptsize
\begin{center}
\setlength{\extrarowheight}{0.2pt}
  \begin{tabular}{|l|l||l|l|l} 
    \hline
    \bf Parameter & \bf Symbol & \bf Parameter & \bf Symbol \\ 
    \hline
    Power consumption in the receive (transmit) path & $P_{\rm rxr}$ ($P_{\rm txr}$) & RF transmit power; maximum transmit power & $P_T$; $P_{T_{\max}}$ \\
    \hline
    Power consumption of the receive (transmit) block & $P_{\rm RX}$ ($P_{\rm TX}$) & RF receive power; mean total received power & $P_R$; $\overline{P}_R$ \\
    \hline
    Power consumption of the local oscillator (LO) & $P_{\rm LO}$ & Path loss exponent & $\alpha$ \\
    \hline
    Receiver overhead power & $P_{O}$ & Average receive (transmit) time in one cycle & $\bar{t}_{\rm rx}$ ($\bar{t}_{\rm tx}$)\\
    \hline
    Power consumption of the power amplifier PA & $P_{\rm PA}$ & Transmission/reception time slot duration & $\Delta t$ \\    
    \hline
PA efficiency & $\eta$ & Receiver (transmitter) activity factors & $\beta_R$ ($\beta_T$)\\
    \hline

  \end{tabular}
\end{center}
\caption{Notation.}
\label{table:tab1}
\end{table}

%%%%%%%%%%%%%%%%%%%%%%%%%%%%%%%%%%%%%%%%%%%%%%%%%%%%%%%%%%%%%%%%%%%%%%%%%%%%%%%%%%%%%%%%%%%%%%%%%%%%%%%%%%%%%%%%%%%%%%%%%%
\section{Single-Stage Energy Density Optimization}
\label{single-stage}
The main focus of this section is to model the average total energy density for a single-stage data aggregation scheme, which is a single hop energy model incorporating the data aggregation and transceiver model described in Sect. \ref{systemmodel}. This model paves the way for understanding the multi-stage energy model to be discussed in Sect. \ref{multi-stage}. 

For the single-stage model, recall that initially a density $\lambda_a<\lambda/2$ of the devices will be independently selected as aggregators, and a density of $\lambda_u=\lambda-\lambda_a$ will be the transmitters. Then, the set of aggregators, $\Psi_a$, collects the data from the remaining devices, $\Psi_u$, based on nearest aggregator association, and each aggregator might have multiple devices assigned to it. However, multiple transmissions to an aggregator at a particular time slot are not allowed, hence the devices are scheduled based on TDMA. Once all devices complete their transmissions, the aggregators also incorporate their payloads, and then transmit the whole data to the BS.

We assume perfect channel inversion power control, under which the received power at the BS from any device is unity. Then, the average total uplink power is given by Theorem \ref{maintheo}.
\begin{theo}
\label{maintheo}
The mean total uplink power of the devices in the Voronoi cell of the typical aggregator, i.e., the average of the total power consumption of the PAs, is given by
\begin{eqnarray}
\label{PAthm}
P(\lambda_a)=\frac{\pi\lambda_u \overline{P}_T}{\eta(\pi\lambda_a )^{1+\alpha/2}} \gamma\left(\frac{\alpha}{2}+1,\lambda_a\pi r_c^2\right)+\frac{\lambda_u P_{T_{\max}}}{\eta\lambda_a} e^{-\lambda_a\pi r_c^2},
\end{eqnarray}
where $r_c=\big( {P_{T_{\max}}}/{\overline{P}_T}\big)^{1/\alpha}$ is defined as the critical distance. 
\end{theo}
\begin{proof}
Proof is given in Appendix \ref{App:AppendixA}, where the final result can be obtained by incorporating the maximum power constraint into the mean additive characteristic associated with the typical cell of the access network model provided in \cite[Ch. 4.5]{BaccelliBook1}.
 \qedhere
\end{proof}

\begin{cor}
If the maximum transmit power is unbounded, the mean total uplink power becomes
\begin{eqnarray}
\lim_{P_{T_{\max}} \to \infty} P(\lambda_a)=\frac{\pi\lambda_u \overline{P}_T }{\eta(\pi\lambda_a)^{1+\alpha/2}}\Gamma\left(\frac{\alpha}{2}+1\right).\nonumber
\end{eqnarray}	
\end{cor}

Let $\mathrm{N_a}$ be a random variable that denotes the number of devices within the Voronoi cell of a typical aggregator and its average can be obtained as
\begin{eqnarray}
\label{Na}
\mathbb{E}\left[\mathrm{N_a}\right]=\mathbb{E}_{\Psi_a}^0{\sum\limits_{n}{1_{\bar{X}_n\in V_0}}}=2\pi\lambda_u \int_{r>0}{ \! e^{-\lambda_a \pi r^2} r\, \mathrm{d}r}=\frac{\lambda_u}{\lambda_a},
\end{eqnarray}
where $\mathbb{E}_{\Psi_a}^0$ is the expectation with respect to the Palm probability conditioned on $0 \in \Psi_a$ \cite{BaccelliBook1}.

We denote the number of devices within the Voronoi cell of a typical aggregator and within distance $d$ by $\mathrm{N_a}(d)$, and its mean can be obtained as
\begin{eqnarray}
\label{Nad}
\mathbb{E}\left[\mathrm{N_a}(d)\right]=\mathbb{E}_{\Psi_a}^0{\sum\limits_{n}{1_{\bar{X}_n\in V_0 \cap \mathrm{B}(0,d))}}}=2\pi\lambda_u \int_{0}^{d}{ \! e^{-\lambda_a \pi r^2} r\, \mathrm{d}r}=\frac{\lambda_u}{\lambda_a}(1-e^{-\lambda_a\pi d^2}),
\end{eqnarray}
where $\mathrm{B}(x,\norm{x})$ is the closed ball centered at $x$ and of radius $\norm{x}$.

\begin{lem}\label{meanPR}
	The mean total received power is given by
	\begin{eqnarray}
	\overline{P}_R=\frac{\lambda_u}{\lambda_a}\left(1-\exp(-\pi\lambda_a r_c^2)\right) \overline{P}_T+\frac{\pi\lambda_u}{(\pi\lambda_a)^{-\alpha/2+1} } P_{T_{\max}}\Gamma\left(1-\frac{\alpha}{2},\pi\lambda_ar_c^2\right).
	\end{eqnarray}
\end{lem}
\begin{proof}
Proof is given in Appendix \ref{App:AppendixB}. \qedhere
\end{proof}

We now define the average energy terms as functions of the aggregator and device densities.

\begin{defi}
{\bf Average total energy in a Voronoi cell of an aggregator.} This is the average total energy consumption due to the transmissions within the coverage of the typical aggregator for a fixed payload per device. The average total energy in a typical cell is given by
\begin{align}
\label{EV}
\mathcal{E}_V(\lambda_a)=E_R(\lambda_a)+E_T(\lambda_a).
\end{align}
\end{defi}

\begin{defi}
{\bf Average total energy density.} This is the average total energy consumption per unit area per unit time for a fixed payload per device, and can be found by scaling the average energy density per Voronoi cell by the number of Voronoi cells per area. The average number of Voronoi cells per unit area is given by $\mathbb{E}\left[N_v\right]={\lambda}/{(\lambda_u/\lambda_a+1)}=\lambda_a$.

The dissipated energy density, i.e., the energy consumption per unit area, is defined as
\begin{align}
\label{problem}
\mathcal{E}(\lambda_a)=\lambda_a\mathcal{E}_V(\lambda_a)=\lambda_a\left(E_R(\lambda_a)+E_T(\lambda_a)\right).
\end{align}
\end{defi}

Our objective is to find the optimal value of $\lambda_a$ that minimizes the total energy consumption per unit area. As $\lambda_a$ increases, the total number of aggregators and their total energy consumption to receive data will increase, and as $\lambda_a$ decreases, the distance between the aggregators and the energy requirement of the typical transmitter will increase. We modify (\ref{ER}) as
\begin{eqnarray}
\label{ERandETlambdaa}
E_R(\lambda_a)=\mathbb{E}[\Delta t]\left(P_{\rm LO}+ P_{\rm RX}+P_{\rm O}\right), \quad
E_T(\lambda_a)=\bar{t}_{\rm tx}\left[\mathbb{E}\left[\mathrm{N_a}^2\right]P_{\rm LO}+\mathbb{E}\left[\mathrm{N_a}\right]P_{\rm TX}+P(\lambda_a)\right],
\end{eqnarray}
where we recall $\bar{t}_{\rm rx}=\bar{t}_{\rm tx}\mathbb{E}\left[\mathrm{N_a}\right]=\Delta t$ and $\mathbb{E}\left[\mathrm{N_a}^2\right]=\frac{\lambda_u}{\lambda_a}+\frac{4.5}{3.5}\left(\frac{\lambda_u}{\lambda_a}\right)^2$ is the second moment of the number of devices per aggregator\footnote{We derive $\mathbb{E}\big[\mathrm{N_a}^2\big]$ in Sect. \ref{rate} using an approximate model for the Voronoi cell areas detailed in \cite{JaraiSzabo2008}.}. 

Any aggregator device aggregates data from multiple devices. Hence, the average number of devices per aggregator, i.e., $\mathbb{E}\left[\mathrm{N_a}\right]={\lambda_u}/{\lambda_a}={(1-\gamma)}/{\gamma}$, is always greater than 1. Thus, the optimal fraction of aggregators that minimizes the overall energy density should satisfy $\gamma< 0.5$.

%%%%%%%%%%%%%%%%%%%%%%%%%%%%%%%%%%%%%%%%%%%%%%%%%%%%%%%%%%%%%%%%%%%%%%%%%%%%%%%%%%%%%%%%%%%%%%%%%%%%%%%%%%%%%%%%%%%%%%%%%%
\section{Multi-Stage Energy Density Optimization}
\label{multi-stage}
We extend the model in Sect. \ref{single-stage} to the case where the aggregation process is repeated $K>1$ times before the data is eventually transmitted to the BS. Let $E_{R,k}(\lambda_a(k))$ and $E_{T,k}(\lambda_a(k))$ be the average energies required for reception and transmission in the coverage of an aggregator and $\lambda_a(k)$ be the aggregator density at stage $k$. The average total energy density at stage $k$ is
\begin{eqnarray}
\mathcal{E}_k(\lambda_a(k))=\lambda_a(k)(E_{R,k}(\lambda_a(k))+E_{T,k}(\lambda_a(k))).
\end{eqnarray}  
The initial total density of the devices is given by $\lambda$. The initial process with density $\lambda$ is independently thinned to obtain the set of aggregator processes with density $\lambda_a(k)$ for stage $k$ to form a disjoint set of aggregators for each stage. The density of aggregators at stage $k$ is denoted by $\lambda_a(k)=\lambda\gamma^{k}$, for $k\in \{1, \hdots, K-1\}$. The total density of active devices that can be either transmitter and receiver at stage $k$ is given by $\lambda(k)=\lambda_u(k)+\lambda_a(k)$.

The transmitter device density of the first stage is found by subtracting the total density of aggregators from the initial density of the device process as $\lambda_u(1)=\lambda-\sum\nolimits_{k=1}^{K-1}{\lambda_a(k)}$. By the end of the first stage, the devices with a density of $\lambda_u(1)$ will transmit their payload to the aggregator process with a density of $\lambda_a(1)$. Next, at the second stage, the aggregators of the first stage form the new transmitting device process, i.e., $\lambda_u(2)=\lambda_a(1)$, and these devices transmit to the devices forming the set of second stage aggregator process with density $\lambda_a(2)$. The aggregation process can be extended to $k>2$ stages in a similar manner by letting $\lambda_u(k)=\lambda_a(k-1)$ for $k\geq 2$. In Fig. \ref{fig-model}, we illustrate a three-stage model, where the aggregator process of stage $k$, i.e., $\Psi_a(k)$, is obtained by independently thinning $\Psi$ with probability $0.4^k$, i.e., $\Psi_a(k)$ has a density of $0.4^k\lambda$.

\begin{figure}[t!]
\centering
\includegraphics[width=\columnwidth]{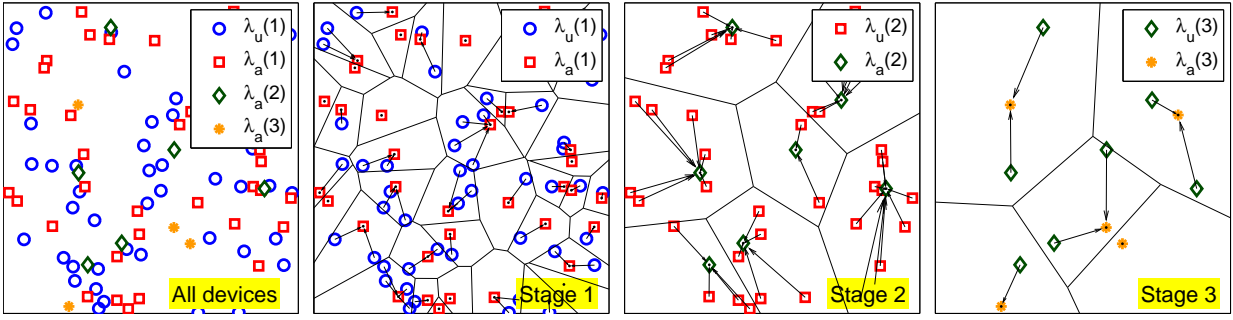}
\caption{\footnotesize{A three-stage aggregation scheme with PPP($\lambda=1$) device process. The aggregator processes (stage 1: $\Box$, stage 2: $\Diamond$, stage 3: $*$) are obtained by thinning the original PPP, and have densities $\lambda_a(k)=0.4^k$ for stages k=1:3. Density for the initial transmitter process (stage 1: $\circ$) is $\lambda_u(1)=\lambda-\sum\nolimits_{k=1}^{3}{\lambda_a(k)}$, and for the transmitter processes for later stages are $\lambda_u(k)=\lambda_a(k-1)$ for $k>2$, respectively. Last stage aggregators correspond to the BSs.}}	
\label{fig-model}
\end{figure}

Generalizing (\ref{ERandETlambdaa}), the average energy consumptions per unit area at stage $k$ are given by
\begin{eqnarray}
\label{energyandcoverage}
&E_{R,k}(\lambda_a(k))=\Delta t(k)\mathcal{P}_{\mathrm{cov}}(k-1)\left(P_{\rm LO}+P_{\rm O}+ P_{\rm RX}\right)\\
\label{TXenergyandcoverage}
&E_{T,k}(\lambda_a(k))=\bar{t}_{\rm tx}(k)\mathcal{P}_{\mathrm{cov}}(k-1)\Big{[}\mathbb{E}\left[\mathrm{N_a}(k)^2\right]P_{\rm LO}+\mathbb{E}\left[\mathrm{N_a}(k)\right]P_{\rm TX}+P(\lambda_a(k))\Big{]},
\end{eqnarray}
where $\Delta t(k)=\bar{t}_{\rm tx}(k)\mathbb{E}\left[\mathrm{N_a}(k)\right]$, and $\mathcal{P}_{\mathrm{cov}}(k)$ is the joint SIR coverage probability up to and including stage $k$ at a given threshold. SIR coverage probability gives the fraction of successful transmissions with respect to an SIR threshold, and the actual average energy at stage $k$ is determined by the fraction of successful transmissions up to stage $k$, i.e., the joint SIR coverage $\mathcal{P}_{\mathrm{cov}}(k-1)$. Hence, the current stage energy consumption scales with the fraction of the successful transmissions. For the devices in outage, no additional energy expenditure is incurred in the consecutive stages. In this paper, we interchangeably use $\lambda_a$ (or $\lambda_a(k)$), $\lambda_u$ (or $\lambda_u(k)$), $\mathrm{N_a}$ (or $\mathrm{N_a}(k)$), and $\Delta t(k)$ (or $\Delta t$) when the stage index is insignificant or clear from context.

\subsection{An Upper Bound on the Energy Density}\label{UBenergydensity}
Let $\pmb{\lambda_a}(K)=\begin{bmatrix}\lambda_a(1) & \lambda_a(2) & \cdots & \lambda_a(K)\end{bmatrix}^\intercal$. Following the aggregation procedure, an upper bound for the average total energy consumption density over $K$ stages is found by ignoring the SIR coverage probability as $\mathcal{E}_{\mathrm U}(\pmb{\lambda_a}(K))=\big. \sum\nolimits_{k=1}^{K}{\mathcal{E}_{k}(\lambda_a(k))} \big\vert_ {\mathcal{P}_{\mathrm{cov}}(k)=1, \forall k}$ that is equivalent to 
\begin{multline}
\label{energyupper}
\mathcal{E}_{\mathrm U}(\pmb{\lambda_a}(K))
=\sum\nolimits_{k=1}^{K}\bar{t}_{\rm tx}(k)\Big[\lambda_u(k)P_C+\lambda_a(k)P(\lambda_a(k))+\Big(\lambda_u(k)+\frac{4.5}{3.5}\frac{\lambda_u(k)^2}{\lambda_a(k)}\Big)P_{\rm LO}\Big],
\end{multline}
where $P_C=P_{\rm TX}+P_{\rm RX}+P_{\rm LO}+P_{\rm O}$. We simplify (\ref{energyupper}) by letting $\bar{t}_{\rm tx}(1)=1$ as this scaling does not affect the result of the optimization formulation.

The slot duration at stage $k$ is $\Delta t(k)=\bar{t}_{\rm rx}(k)$ and the duration that the transmit block is on is $\bar{t}_{\rm tx}(k)$. The relation between the average transmit times are as $\bar{t}_{\rm tx}(k+1)=\bar{t}_{\rm tx}(k)\mathbb{E}\left[\mathrm{N_a}(k)\right]$. The transmit duration at slot $k$ is proportional to the mean total number of bits to be transmitted as $\bar{t}_{\rm tx}(k)\propto \mathbb{E}\left[\prod\nolimits_{i=1}^{k-1}{\mathrm{N_a}(i)}\right]\stackrel{(a)}{\geq} \prod\nolimits_{i=1}^{k-1}{\mathbb{E}\left[\mathrm{N_a}(i)\right]}$. Note that $\mathrm{N_a}(i)$'s for $i\in\{1,\hdots, k-1\}$ are dependent random variables because of the correlation between the Voronoi cell areas of the aggregators at subsequent stages \cite{Yu2013}. In fact, the inequality $(a)$ follows from the observation that $\mathrm{N_a}(i)$'s are positively associated, i.e., when the values of one variable tend to increase as the values of the other variable increase as can be seen from Fig. \ref{fig-cor}.

\begin{figure}[t!]
	\centering
	\includegraphics[width=0.5\linewidth]{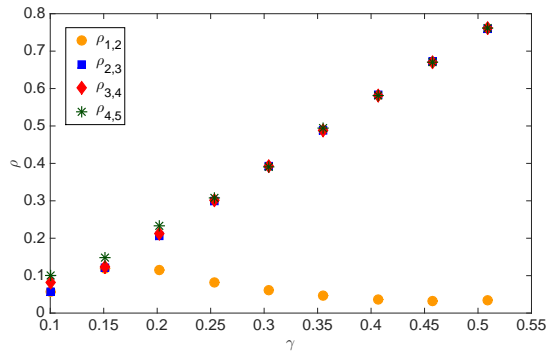}
	\caption{\footnotesize{Numerical correlation over multiple stages.}}
	\label{fig-cor}
\end{figure}

We denote the correlation between the number of devices within the Voronoi cells of the typical aggregators of the consecutive stages $k$ and $k+1$, i.e., the correlation between $\mathrm{N_a}(k)$ and $\mathrm{N_a}(k+1)$, by $\rho_{k,k+1}$, which is expressed as
\begin{eqnarray}
\rho_{k,k+1}=\frac{\mathbb{E}[\mathrm{N_a}(k)\mathrm{N_a}(k+1)]-\mathbb{E}[\mathrm{N_a}(k)]\mathbb{E}[\mathrm{N_a}(k+1)]}{\sqrt{\mathrm{Var}[\mathrm{N_a}(k)]\mathrm{Var}[\mathrm{N_a}(k+1)]}}.
\end{eqnarray}
To the best of our knowledge, there is no analytical model in the literature for the correlation between the number of stages for the proposed data aggregation model. However, we can observe the trend of the correlation numerically. We illustrate how the correlation behaves for $K=4$ number of stages in Fig. \ref{fig-cor}. From Table \ref{table:tab3}, the average number of devices per aggregator is $\mathbb{E}[\mathrm{N_a}(1)]=(1-\sum\nolimits_{k=1}^{K-1}\gamma^k)/\gamma$ for stage 1 and $\mathbb{E}[\mathrm{N_a}(k)]=\gamma^{-1}$ for stages $k=2,\hdots, K$. As $\gamma$ increases, $\mathbb{E}[\mathrm{N_a}(1)]$ decreases with a rate $\gamma^{-2}[1+\sum\nolimits_{k=2}^{K-1}{(k-1)\gamma^k}]$ which is faster than $\mathbb{E}[\mathrm{N_a}(k)]$ for $k>2$ with a decrease rate of $\gamma^{-2}$. Therefore, $\mathrm{N_a}(1)$ dominates the trend of $\rho_{1,2}$ and thus, we observe a different behavior for $\rho_{1,2}$ compared to $\rho_{i,i+1}$ for $i>1$, where the correlation is similar because the device and aggregator densities for higher stages are obtained with the same scaling parameter $\gamma$. Correlation satisfies $|\rho_{k,k+1}|<0.4$ for $\gamma\in (0,0.3)$, which is sufficiently small given large number of stages. Thus, letting $\gamma\in (0,0.3)$ for analytical tractability, and assuming inter-stage independence among $\mathrm{N_a}(i)$'s, we can modify the expectation of the product as
\begin{eqnarray}
\bar{t}_{\rm tx}(k)\propto\prod\nolimits_{i=1}^{k-1}{\mathbb{E}\left[\mathrm{N_a}(i)\right]}=\prod\nolimits_{i=1}^{k-1}{\lambda_u(i)/\lambda_a(i)}.
\end{eqnarray} 

\begin{remark}
If there are $K$ stages in total, at the last stage, all the aggregators of the previous stage (stage $K-1$) with density $\lambda_u(K)=\lambda_a(K-1)$ are transmitting to the BSs with density $\lambda_{\mathrm{BS}}$. Thus, at stage $K$, $\lambda_a(K)=\lambda_{\mathrm{BS}}$. Therefore, the fraction of aggregators at the last stage is
\begin{eqnarray}
\gamma(K)
=\lambda_{\mathrm{BS}}(\lambda\gamma^{K-1}+\lambda_{\mathrm{BS}})^{-1},
\end{eqnarray}
which is incorporated into the numerical analysis to find the minimizing aggregator fraction. 
\end{remark}

\subsection{Average Energy Density with an SIR Coverage Constraint}
\label{MeanEnergyDensity}
This section mainly concentrates on the average energy consumption in the case of SIR outage. The average energy density of the proposed aggregator model is found by incorporating the SIR coverage probability of the typical transmitting device. If the received SIR at any stage $k$ cannot exceed the threshold $T$, then the transmission is unsuccessful. Therefore, devices with a density of $\mathcal{P}_k(T)\lambda_u(k)$ will be successful at stage $k$, i.e., out of $\mathrm{N_a}(k)$ transmitting devices served per aggregator, on average, only $\mathcal{P}_k(T)\mathbb{E}\left[\mathrm{N_a}(k)\right]$ of them will successfully transmit their data, and the rest of the devices will not meet the minimum SIR requirement despite consuming energy. Therefore, the performance of the current stage depends on the previous stages. 

Incorporating the SIR coverage, the average total energy density in (\ref{energyupper}) is modified as
\begin{eqnarray}
\label{SequentialTotalPower}
\mathcal{E}(\pmb{\lambda_a}(K))=\sum\limits_{k=1}^{K}\mathcal{P}_{\mathrm{cov}}(k-1)\bar{t}_{\rm tx}(k)\Big{[}\lambda_u(k)P_C+\lambda_a(k)P(\lambda_a(k))+\big(\lambda_u(k)+\frac{4.5}{3.5}\frac{\lambda_u(k)^2}{\lambda_a(k)}\big)P_{\rm LO}\Big{]},
\end{eqnarray}
where $\mathcal{P}_{\mathrm{cov}}(k)$ denotes the joint SIR coverage probability for $k$ stages and $\mathcal{P}_{\mathrm{cov}}(0)=1$, where we do not consider the effect of SIR coverage probability as all the devices of first stage transmit even though they might not be successful. In Sect. \ref{coverageprobability}, we consider sequential and parallel modes with the inter-stage independence assumption, i.e., the joint SIR coverage becomes $\mathcal{P}_{\mathrm{cov}}(k)=\prod_{i=1}^{k}\mathcal{P}_i(T)$, where the dependence on $T$ is inherent and omitted to keep the notation concise.  

Next, we formulate the energy density optimization problem assuming that SIR coverage probability of each device is unity, i.e., $\mathcal{P}_{\mathrm{cov}}(k)=1$ for $k\in\{1,\hdots, K\}$. Note that $\mathcal{P}_{\mathrm{cov}}(k)$ depends on the communication scheme, the network model parameters, such as device density, fading distribution, and the path loss, and is determined independently from the energy model. This section mainly concentrates on the energy model, while the SIR coverage and rate models -- and hence, $\mathcal{P}_{\mathrm{cov}}(k)$'s for various transmission schemes -- will be discussed extensively in Sects. \ref{coverageprobability}-\ref{rate}. In Sect. \ref{Performance}, we combine the energy results with the coverage characteristics for the successive and parallel modes, to evaluate the energy efficiencies of the proposed models.

\subsection{Energy Density Optimization Problem}
The following optimization problem minimizes the average total energy density over $K$ stages:
\begin{equation}
\begin{aligned}
\label{eq:M2M-opt}
\underset{\pmb{\lambda_a}(K)}{\min} &\quad\mathcal{E}(\pmb{\lambda_a}(K))\\
\textrm{s.t.} 
& \quad \lambda_u(k)= \lambda_a(k-1),\quad k\in \{2,\hdots, K\}\\
&\quad \bar{t}_{\rm tx}(k)=\bar{t}_{\rm tx}(k-1)\mathbb{E}\left[\mathrm{N_a}(k-1)\right],\quad k\in \{2,\hdots, K\}\\
&\quad \lambda_u(1)=\lambda-\lambda\sum\nolimits_{k=1}^{K-1}{\gamma^k},
\end{aligned}
\end{equation} 
where the aggregator density is $\lambda_a(k)= \lambda\gamma^{k}$ for $1\leq k<K$, and $\lambda_a(K)=\lambda_{\rm BS}$ for stage $K$.

For the optimization formulation in (\ref{eq:M2M-opt}), and for $K\geq 2$, we let ${\bar{\gamma}}_K=\sum_{k=1}^{K-1}{\gamma^k}$. The important design parameters of the total energy density are tabulated in Table \ref{table:tab3}. Using (\ref{PAthm}) and the design parameters in Table \ref{table:tab3}, the mean total uplink power for different stages is given by
\begin{eqnarray}
P(\lambda_a(k))=
\begin{cases}
\frac{(1-{\bar{\gamma}}_K) {\gamma}^{-(1+\alpha/2)} \overline{P}_T}{\eta(\pi\lambda)^{\alpha/2}}\Gamma\left(\frac{\alpha}{2}+1,\lambda\gamma\pi r_c^2\right)+\frac{(1-{\bar{\gamma}}_K)P_{T_{\max}}}{\eta\gamma}e^{-\lambda\gamma\pi r_c^2}, \quad k=1\\
\frac{\gamma^{-(1+k\alpha/2)}\overline{P}_T}{\eta(\pi\lambda)^{\alpha/2}}\Gamma\left(\frac{\alpha}{2}+1,\lambda\gamma^k\pi r_c^2\right)+\frac{P_{T_{\max}}}{\eta\gamma}e^{-\lambda\gamma^k\pi r_c^2}, \quad 2\leq k\leq K-1\\
\frac{\lambda\gamma^{K-1} \overline{P}_T}{\eta{\lambda_{\rm BS}}^{1+\alpha/2}\pi^{\alpha/2}}\Gamma\left(\frac{\alpha}{2}+1,\lambda_{\rm BS}\pi r_c^2\right)+\frac{\lambda\gamma^{K-1}P_{T_{\max}}}{\eta\lambda_{\rm BS}}e^{-\lambda_{\rm BS}\pi r_c^2}, \quad k=K
\end{cases}.
\end{eqnarray}

\begin{prop}
\label{simple}
Combining (\ref{SequentialTotalPower}) with the constraints in (\ref{eq:M2M-opt}) for unit SIR coverage probability, we define $c_k(\gamma)$ to be the total energy density at stage $k\in\{1,\hdots,K\}$. Hence, $c_k(\gamma)$ is given by
\begin{eqnarray}
\label{ckequation}
c_k(\gamma)=
\begin{cases}
\lambda\Big[\left(1-{\bar{\gamma}}_K\right)P_C+\gamma P(\lambda\gamma)+\big[ \left(1-{\bar{\gamma}}_K\right)+\frac{4.5}{3.5}\frac{\left(1-{\bar{\gamma}}_K\right)^2}{\gamma}\big]P_{\rm LO}\Big], \quad k=1\\
\lambda\left(1-{\bar{\gamma}}_K\right)\Big[P_C+\gamma P(\lambda\gamma^k)+\big[1+\frac{4.5}{3.5}\gamma^{-1}\big]P_{\rm LO}\Big], \quad 2\leq k\leq K-1\\
\lambda\left(1-{\bar{\gamma}}_K\right)\Big[P_C+\frac{\lambda_{\rm BS}}{\lambda\gamma^{K-1}}P(\lambda_{\rm BS})+\big[1+\frac{4.5}{3.5}\frac{\lambda}{\lambda_{\rm BS}}\gamma^{K-1}\big]P_{\rm LO}\Big], \quad k=K
\end{cases}.
\end{eqnarray}

Incorporating the constraints, and using (\ref{ckequation}), the objective function of (\ref{eq:M2M-opt}) is equivalent to 
\begin{eqnarray}
\label{optimizationsol}
\mathcal{E}(\pmb{\lambda_a}(K))= \sum\limits_{k=1}^{K}\mathcal{P}_{\mathrm{cov}}(k-1)c_k(\gamma).
\end{eqnarray}
Similarly, the energy upper bound $\mathcal{E}_{\mathrm U}(\pmb{\lambda_a}(K))$ can be minimized by evaluating (\ref{optimizationsol}) at $\mathcal{P}_{\mathrm{cov}}(k)=1$ $\forall k$, and taking its derivative with respect to $\gamma$.
\end{prop}

\begin{prop}\label{Propk}
$c_k(\gamma)$ is an increasing function of $k$ for $1\leq k\leq K-1$, and $c_K(\gamma)$ is a decreasing function of total number of stages $K$.
\end{prop}

\begin{proof}
Note that $\gamma^{-k}$ increases with $k$, and $P(\gamma^{k}\lambda)$ is an increasing function of $k$ for $2\leq k\leq K-1$. Using (\ref{ckequation}), it is easy to show that $c_k(\gamma)<c_{k+1}(\gamma)$ for $1\leq k\leq K-2$ for fixed $\gamma$. Note also that $P(\lambda_{\rm BS})$ decreases, and $P(\lambda_{\rm BS})/\gamma^{K-1}$ does not change with $K$, but $\gamma^{K-1}P_{\rm LO}$ decreases with $K$, hence it is trivial to show that $c_K(\gamma)$ is decreasing in $K$. 
\end{proof}
Proposition \ref{Propk} implies that the energy density is higher for higher order stages. On the other hand, as the number of stages increases, the energy density at the last stage decreases. The number of stages cannot be increased arbitrarily. Next, to capture the energy consumption tradeoff between the stages, we investigate the effect of $\gamma$ on the total energy density.

\begin{remark}\label{rem}
Using (\ref{ckequation}), we can easily observe that $c_{k}(\gamma)$ is decreasing in $\gamma$ for $1\leq k\leq K-1$. Note also that $c_K(\gamma)$ is given by the following sum
\begin{eqnarray}
c_K(\gamma)=\lambda\left(1-{\bar{\gamma}}_K\right)\left[P_C+\frac{\lambda_{\rm BS}}{\lambda\gamma^{K-1}}P(\lambda_{\rm BS})+P_{\rm LO}\right]+\lambda\left(1-{\bar{\gamma}}_K\right)\Big[\frac{4.5}{3.5}\frac{\lambda}{\lambda_{\rm BS}}\gamma^{K-1}\Big]P_{\rm LO}.\nonumber
\end{eqnarray} 
We observe that the first term decays with $\gamma$ as $1-{\bar{\gamma}}_K$ decreases with $\gamma$, and $P(\lambda_{\rm BS})/\gamma^{K-1}$ is invariant to $\gamma$. The second term is not strictly a decreasing function of $\gamma$. On the other hand, it is strictly increasing for $K\geq 2$ and $\gamma\in (0,a)$, where $a\in (0,0.5)$ and decreasing in $K$. As the second term is scaled by $\lambda/\lambda_{\rm BS}$, which is typically very high for the operating regime of M2M, it determines the trend of total energy.
\end{remark}

\begin{table}[t!]\scriptsize
\centering
\setlength{\extrarowheight}{4pt}
  \begin{tabular}{l||lll}
     & $k=1$ & $ 2 \leq k\leq K-1$ & $k=K$ \\ 
    \hline
    \hline
    $\lambda_u(k)$ & $\lambda(1-{\bar{\gamma}}_K)$ & $\lambda\gamma^{k-1}$ & $\lambda\gamma^{K-1}$\\   
    $\lambda_a(k)$ & $\lambda\gamma$ & $\lambda\gamma^k$ & $\lambda_{\rm BS}$\\
    $\mathbb{E}\left[\mathrm{N_a}(k)\right]$ & $(1-{\bar{\gamma}}_K)/\gamma$ & $\gamma^{-1}$ & $\lambda\gamma^{K-1}/\lambda_{\rm BS}$\\ 
    $\bar{t}_{\rm tx}(k)$ & 1 & $(1-{\bar{\gamma}}_K)/\gamma^{k-1}$ &$(1-{\bar{\gamma}}_K)/\gamma^{K-1}$
  \end{tabular}
\caption{\footnotesize{Design parameters.}}
\label{table:tab3}
\end{table}

\begin{prop}
\label{mainprop}
Let $\gamma^{(K)*}$ be the minimizer of the energy optimization problem in (\ref{optimizationsol}) with the assumption\footnote{Note that it is trivial to extend this result for $\mathcal{P}_{\mathrm{cov}}(k)<1$ because $c_k(\gamma)$'s are invariant to the coverage probability.} of $\mathcal{P}_{\mathrm{cov}}(k)=1$ $\forall k$, and $c_K(\gamma)$ is increasing in $\gamma$. Then, $\gamma^{(K)*}$ is a decreasing function of the total number of stages $K$.
\end{prop}
\begin{proof}
See Appendix \ref{App:AppendixC}.
\qedhere
\end{proof}

This section concentrates on the energy efficiency of M2M assuming perfect coverage. We now study the SIR coverage of the hierarchical M2M model for different transmission schemes.

%%%%%%%%%%%%%%%%%%%%%%%%%%%%%%%%%%%%%%%%%%%%%%%%%%%%%%%%%%%%%%%%%%%%%%%%%%%%%%%%%%%%%%%%%%%%%%%%%%%%%%%%%%%%%%%%%%%%%%%%%%
\section{SIR Coverage Probability}
\label{coverageprobability}
The data aggregation models in Sects. \ref{single-stage} and \ref{multi-stage} do not incorporate the fact that the transmissions are not always successful. In this section, we generalize this to a coverage-based model, where the transmission is successful if the SIR of the device is above a threshold. We derive the probability of coverage in the uplink for the proposed data aggregation model in Sects. \ref{single-stage} and \ref{multi-stage}. We assume that the tagged aggregator and tagged devices experience Rayleigh fading, and we ignore shadowing. The transmit power at the typical node at distance $r$ from its aggregator is $P_T(r)=\min\{P_{T_{\max}}, \overline{P}_T r^{\alpha}\}$, and the received power is $P_R(r)=g P_T(r) r^{-\alpha}$, where $g \sim \exp(1)$. 

Orthogonal access is assumed in the uplink and at any given resource block, there is at most one device transmitting in each cell. Let $\Psi_u$ be the point process denoting the location of devices transmitting on the same resource as the typical device. The uplink SIR of the typical device $x\in\Psi_u$ located at distance $\norm{x}$ from its associated BS (aggregator)  on a given resource block is
\begin{eqnarray}
\mathrm{SIR}=\left.\frac{P_R(r)}{I_r}\right\vert_{r=\norm{x}}=\frac{g \min\{P_{T_{\max}}\norm{x}^{-\alpha}, \overline{P}_T \}}{\sum\limits_{z\in \Psi_u\backslash \{x\}}g_z \min\{P_{T_{\max}}, \overline{P}_T R_z^{\alpha} \} D_z^{-\alpha}},
\end{eqnarray}
where $R_z$ and $D_z$ denote the distance between the transmitter aggregator pair and the distance between the interferer and the typical aggregator, respectively. The random variable $g_z\sim g$ is the small-scale iid fading parameter due to interferer $z$.

\begin{assu}\label{distanceassumption}
	The actual distribution of $R_z$ is very hard to characterize due to the randomness both in the area of the Voronoi cell of the aggregator and in the number of the devices it serves. Therefore, we approximate it by the distance of a randomly chosen point in $\mathbb{R}^2$ to its closest aggregator and hence it can be approximated by a Rayleigh distribution \cite{Novlan2013}:
	\begin{eqnarray}
	\label{distancedensity}
	f_{R_z}(r_z)=({r_z}/{\sigma^2}) e^{-r_z^2/2\sigma^2},\quad r_z\geq 0, \quad \sigma=\sqrt{1/(2\pi\lambda_a)}.
	\end{eqnarray}
\end{assu}

The uplink SIR coverage of the proposed system model is given by the following Lemma.
\begin{lem}
\label{maincoveragelemma}
{\bf The uplink SIR coverage with truncated power control:} 
With truncated power control and with minimum average path loss association\footnote{In ``minimum average path loss association", a device associates to an aggregator with minimum path loss averaged over the small-scale fading, i.e., the aggregator has minimum $R_z^{\alpha}$ product among all aggregators.}, the uplink SIR coverage is given by
\begin{eqnarray}
\label{uplink-SIR-truncated}
\mathcal{P}(T)= p\mathcal{L}_{I_r}(T\overline{P}_T^{-1})+\int_{r_c}^{\infty}{\!\mathcal{L}_{I_r}(Tr^{\alpha}P_{T_{\max}}^{-1})f_R(r)\, \mathrm{d}r},
\end{eqnarray}
where $p=1-\exp\big(-\pi\lambda_a r_c^2\big)$, $R$ is Rayleigh distributed with parameter $\sigma=\sqrt{1/(2\pi\lambda_a)}$, and 
\begin{eqnarray}
\label{laplacetransformgeneral}
\mathcal{L}_{I_r}(s)&\approx&
\exp\Big(- \frac{2s}{\alpha-2}\Big((1-e^{-\pi\lambda_ar_c^2}(1+\pi\lambda_ar_c^2))\overline{P}_TC_{\alpha}(s\overline{P}_T)\nonumber\\
&+&(1-p)\pi\lambda_aP_{T_{\max}}\mathbb{E}_{R_z}\left[\left.R_z^{2-\alpha}C_{\alpha}(sP_{T_{\max}}R_z^{-\alpha})\right\vert R_z>r_c\right]\Big)\Big)
\end{eqnarray}
denotes the Laplace transform of the interference where $C_{\alpha}(T)={_2F_1}\Big(1,1-\frac{2}{\alpha},2-\frac{2}{\alpha},-T\Big)$, and $_2F_1$ is the Gauss-Hypergeometric function.
\end{lem}

\begin{proof}
See Appendix \ref{App:AppendixD}.
\qedhere
\end{proof}

\begin{cor}
	\label{maincoveragelemma2}
	{\bf The uplink SIR coverage with open loop power control \cite{Singh2014}:} 
	With open loop power control and with minimum average path loss association, the uplink SIR coverage is
	\begin{eqnarray}
	\label{uplink-SIR}
	\lim_{P_{T_{\max}} \to \infty}\mathcal{P}(T)\approx\exp{\Big(-\frac{2 T}{\alpha-2}C_{\alpha}(T)\Big)}.
	\end{eqnarray}
\end{cor}
\begin{proof}
	As $P_{T_{\max}} \to \infty$, $r_c=\big({P_{T_{\max}}}/{\overline{P}_T}\big)^{1/\alpha} \to \infty$ and $p \to 1$, which yields the final result.
\end{proof}

\begin{cor}
The uplink SIR coverage is lower bounded by
\begin{eqnarray}
\mathcal{P}^{\rm LB}(T)= p\mathcal{L}_{I_r}^{\rm LB}(T\overline{P}_T^{-1})+\int_{r_c}^{\infty}{\!\mathcal{L}_{I_r}^{\rm LB}(Tr^{\alpha}P_{T_{\max}}^{-1})f_R(r)\, \mathrm{d}r},
\end{eqnarray}	
where $\mathcal{L}_{I_r}^{\rm LB}(s)$ is a lower bound for the Laplace transform of the interference and given as
\begin{eqnarray}
\mathcal{L}_{I_r}^{\rm LB}(s)&\approx&\exp\Big(- \frac{2s}{\alpha-2}\Big((1-e^{-\pi\lambda_ar_c^2}(1+\pi\lambda_ar_c^2))\overline{P}_TC_{\alpha}(s\overline{P}_T)\nonumber\\
&+&P_{T_{\max}}(\pi\lambda_a)^{\alpha/2}\Gamma(2-2/\alpha,\pi\lambda_ar_c^2)\Big)\Big).
\end{eqnarray}
\end{cor}

\begin{proof}		
Noting that $C_{\alpha}(s)=\frac{\frac{\alpha}{2}-1}{s}\int_{1}^{\infty}{ \! \frac{1}{1+s^{-1} t^{\alpha/2}}  \,\mathrm{d}t}<\frac{\frac{\alpha}{2}-1}{s}\int_{1}^{\infty}{ \! \frac{t^{-\alpha/2}}{s^{-1} }  \,\mathrm{d}t} =1$ for $\alpha>2$, we obtain the following upper bound for the conditional expectation in step (g) of (\ref{mainLT}):
\begin{eqnarray}
\label{RcondBound}
(1-p)\mathbb{E}_{R_z}\left[\left.R_z^{2-\alpha}C_{\alpha}(sP_{T_{\max}}R_z^{-\alpha})\right\vert R_z>r_c\right]
\leq (1-p)\mathbb{E}_{R_z}\left[\left.R_z^{2-\alpha}\right\vert R_z>r_c\right],
\end{eqnarray} 
where the RHS can be calculated as
\begin{eqnarray}
\label{RcondG}
(1-p)\mathbb{E}_{R_z}\left[\left.R_z^{2-\alpha}\right\vert R_z>r_c\right]=\int_{r_c^2}^{\infty}{\! v^{1-\alpha/2} \pi\lambda_a e^{-\pi\lambda_a v}\, \mathrm{d}v}
=\frac{1}{(\pi\lambda_a)^{1-\alpha/2}}\Gamma(2-2/\alpha,\pi\lambda_ar_c^2).
\end{eqnarray}
From (\ref{RcondL}), (\ref{RcondBound}) and (\ref{RcondG}), we can lower bound $\mathcal{L}_{I_r}(s)$, which yields the final result.
\end{proof}

\subsection{Coverage Probability and Number of Stages}
\label{number-stages}
The number of multi-hop stages $K$ is mainly determined by the SIR coverage and the {\em distance coverage}, which we define as the probability that the distance between a device and its nearest aggregator is below a threshold. We provide bounds on $K$ using these coverage concepts.
\begin{assu}
\label{assu-1}
{\bf Interstage independence.} The proposed hierarchical aggregation model introduces dependence among the stages of multi-hop communication since each subsequent stage is generated by the thinning of the previous stage. For analytical tractability, we assume that the multi-hop stages are independent of each other\footnote{This assumption is required for the transmission modes described in detail in Sect. \ref{rate}, and is validated in Sect. \ref{Performance}.}. Hence, the transmission is successful if and only if all the individual stages are successful. The success probability over $K$ stages is
\begin{eqnarray}
\mathcal{P}_{\mathrm{cov}}(K)=\mathbb{P}(\mathrm{SIR}_1>T, \hdots, \mathrm{SIR}_K>T)=\prod\nolimits_{k=1}^{K}{\mathcal{P}_k(T)},\nonumber
\end{eqnarray}
where $\mathcal{P}_k(T)$ denotes the coverage probability at stage $k$. With full channel inversion and minimum average path loss association, the uplink SIR coverage is independent of the infrastructure density as given in (\ref{uplink-SIR}). For the case $P_{T_{\max}} \to \infty$, since $\mathcal{P}_k(T)$ in (\ref{uplink-SIR}) is also independent of the device density, and only depends on the threshold $T$ and path loss exponent $\alpha$, is identical for all stages, and denoted by $\mathcal{P}(T)$. 
\end{assu}

\begin{lem}{\bf An upper bound on \em{K} (Sequential mode).}
Given a minimum probability of coverage requirement $\mathcal{P}_{\mathrm{cov}}(K)\!>\!1-\varepsilon$ and an SIR threshold $T$, the number of stages is upper bounded by
\begin{eqnarray}
\label{K-upper}
K_{\mathrm U}=\left\lceil\frac{\log{(1/(1-\varepsilon))}}{-\log(\max_k\mathcal{P}_k(T))}\right\rceil,\quad \rm{and}\quad
\lim_{P_{T_{\max}} \to \infty} K_{\mathrm U}=\left\lceil\frac{\log{(1/(1-\varepsilon))}}{TC_{\alpha}(T)}\left(\frac{\alpha-2}{2}\right)\right\rceil.
\end{eqnarray}
\end{lem}
\begin{proof}
The upper bound is obtained by combining (\ref{uplink-SIR-truncated}) with the condition $\mathcal{P}_{\mathrm{cov}}(K)>1-\varepsilon$ and using the relation $\mathcal{P}_{\mathrm{cov}}(K)=\prod\nolimits_{k=1}^{K}{\mathcal{P}_k(T)}\leq \max_k\mathcal{P}_k(T)^K$. 
\end{proof}

In addition to the SIR outage, since the devices are randomly deployed, any device will be in outage when its nearest aggregator is outside its maximum transmission range. Thus, we also aim to investigate the minimum number of required stages given a distance outage constraint. 

A lower bound on the optimal number of multi-hop stages is given by the following Lemma. 

\begin{lem}\label{LB}{\bf A lower bound on \em{K}.}
The number of stages is lower bounded by
\begin{eqnarray}
\label{K-lower}
K_{\mathrm L}=\Big \lceil{\mathbb{E}[L(\lambda_a)]\mathbb{E}\Big[\frac{1}{\mathrm{N_a}}\Big]{\Big(\frac{P_{R_{\min}}}{P_{T_{\max}}}\Big)}^\frac{1}{\alpha}}\Big \rceil,
\end{eqnarray}
where $L(\lambda_a)$ denotes the total length of the connections, $\mathrm{N_a}$ is the number of devices in the Voronoi cell of a typical aggregator, $P_{T_{\max}}$ is the maximum transmit power and $P_{R_{\min}}$ is the minimum detectable signal power at the receiver.
\end{lem}

\begin{proof}
The idea of the proof is similar to the proof of Theorem \ref{maintheo} (see Appendix \ref{App:AppendixA}), where $f(x)=\norm{x}$. The mean total length of connections, in the Voronoi cell of an aggregator equals 
\begin{eqnarray}
	\label{connectionlength}
	\mathbb{E}[L(\lambda_a)]=\lambda_u \int_{\mathbb{R}^2}{ \! f(x)e^{-\lambda_a \pi \norm{x}^2} \, \mathrm{d}x}=2\pi\lambda_u \int\nolimits_{0}^{\infty}{ \! r^2e^{-\lambda_a \pi r^2} \, \mathrm{d}r}=\frac{\lambda_u}{2\lambda_a^{3/2}}.
\end{eqnarray}
Given a maximum transmitter power constraint, $P_{T_{\max}}$ for each device, the maximum transmission range is given by $\left({P_{T_{\max}}}/{P_{R_{\min}}}\right)^{1/\alpha}$. Dividing $L(\lambda_a)$ by $\mathrm{N_a}$ and taking its expectation with respect to the distribution of $\mathrm{N_a}$, we obtain the mean length of connections\footnote{For tractability, we take expectation over a PPP first to find $\mathbb{E}[L(\lambda_a)]$, and then multiply it by $\mathbb{E}\Big[\frac{1}{\mathrm{N_a}}\Big]$ assuming independence.}, and dividing this ratio by the maximum transmission range, we obtain the desired result.
\end{proof}

The following Lemma provides a lower bound for $K_{\mathrm L}$ that is based on the mean total length of the connections to the typical aggregator and the fraction of aggregators.
\begin{lem}
$K_{\mathrm L}$ is lower bounded by
\begin{eqnarray}
\Big\lceil{\frac{1}{\mathbb{E}\left[\mathrm{N_a}\right]}\frac{(1-\gamma)}{2\lambda^{1/2}{\gamma^{K_{\mathrm L}/2+1}}}\Big(\frac{P_{R_{\min}}}{P_{T_{\max}}}\Big)^{1/\alpha}}\Big\rceil \leq K_{\mathrm L} \leq K_{\mathrm U}.
\end{eqnarray}
\end{lem}
\begin{proof}
Using (\ref{connectionlength}), we can lower bound $K_{\mathrm L}$ as
\begin{eqnarray}
\label{KLproof}
K_{\mathrm L}\geq\Big\lceil{\frac{1}{\mathbb{E}\left[\mathrm{N_a}\right]}\frac{\lambda_u}{2\lambda_a^{3/2}}\Big(\frac{P_{R_{\min}}}{P_{T_{\max}}}\Big)^{1/\alpha}}\Big\rceil,
\end{eqnarray}
where we use the convexity of $1/\mathrm{N_a}$, i.e., $\mathbb{E}\left[\frac{1}{\mathrm{N_a}}\right]\geq \frac{1}{\mathbb{E}[\mathrm{N_a}]}$. Noting that $\lambda_u$ and $\lambda_a$ are functions of $K$, we have the relation $\mathbb{E}[L(\lambda_a)]=\frac{\lambda_u}{2\lambda_a^{3/2}}=\frac{\lambda\gamma^{K-1}(1-\gamma)}{2\left(\lambda\gamma^K\right)^{3/2}}=\frac{(1-\gamma)}{2\lambda^{1/2}{\gamma^{K/2+1}}}$, which is increasing in $K$ since $\gamma\leq 0.5$. Plugging this relation into (\ref{KLproof}), we get the bound.
\end{proof}

%%%%%%%%%%%%%%%%%%%%%%%%%%%%%%%%%%%%%%%%%%%%%%%%%%%%%%%%%%%%%%%%%%%%%%%%%%%%%%%%%%%%%%%%%%%%%%%%%%%%%%%%%%%%%%%%%%%%%%%%%%
\section{Transmission Rate Models}
\label{rate}
For an interference limited network, the rate of the typical device is given by $\mathrm{Rate}=\frac{W}{\mathrm{N_a}}\log{(1+\mathrm{SIR})}$, where $W$ is the total bandwidth of the communication channel, and $\mathrm{N_a}$ is the load at the typical aggregator. The average number of devices served by the typical aggregator is denoted by $\mathbb{E}[\mathrm{N_a}]=\frac{\lambda_u}{\lambda_a}$. Rate coverage is defined as rate exceeding a given threshold, i.e., 
\begin{eqnarray}
\label{coverage-prob}
\mathbb{P}(\mathrm{Rate}>\rho)=\sum\limits_{l=0}^{\infty}{\mathbb{P}\left(\mathrm{SIR}>2^{\frac{\rho \mathrm{N_a}}{W}}-1| \mathrm{N_a}=l\right)\mathbb{P}_{\mathrm{\mathrm{N_a}}}(l)},
\end{eqnarray}
where $\mathbb{P}_{\mathrm{\mathrm{N_a}}}(l)$ is the probability mass function (PMF) of $\mathrm{N_a}$, and is given by the following Lemma.

\begin{lem}
The PMF of the number of devices served per aggregator of stage $k$, i.e., $\mathrm{N_a}(k)$, is
\begin{eqnarray}
\label{NkPMF}
\mathbb{P}_{\mathrm{N_a}(k)}(l)=\frac{G_{\mathrm{N_a}(k)}^{(l)}(0)}{l!}=\frac{3.5^{3.5}(\lambda_u(k)/\lambda_a(k))^{l}}{(3.5+(\lambda_u(k)/\lambda_a(k)))^{3.5+l}}\frac{\Gamma(3.5+l)}{\Gamma(3.5)\Gamma(l+1)},
\end{eqnarray}
where $\lambda_u(k)$ and $\lambda_a(k)$ for $k\geq 1$ are given in Table \ref{table:tab3}.
\end{lem}
\begin{proof}
Normalized distribution function of Voronoi cell areas in 2D can be modeled by \cite{JaraiSzabo2008} as $f(y)=\frac{3.5^{3.5}}{\Gamma(3.5)}y^{\frac{5}{2}}e^{-\frac{7}{2}y}$. Using the densities of the transmitters and the aggregators of stage $k$, the probability generating function (PGF) of the stage $k$ devices in the random area $y$ is \cite{Singh2013}
\begin{eqnarray}
\label{PGF}
G_{\mathrm{N_a}(k)}(z)=\mathbb{E}[z^{\mathrm{N_a}(k)}]=\mathbb{E}[\exp{((\lambda_u(k)/\lambda_a(k)) y (z-1))}],
\end{eqnarray}
where conditioned on $y$, the PGF is of a Poisson random variable $\mathrm{N_a}(k)$ with mean $(\lambda_u(k)/\lambda_a(k)) y$.

Combining (\ref{PGF}) and the pdf $f(y)$, we obtain the PGF of $\mathrm{N_a}(k)$ as
\begin{eqnarray}
\label{PGFofG}
G_{\mathrm{N_a}(k)}(z)=\int\nolimits_{0}^{\infty}{\! \exp{\Big(\frac{\lambda_u(k)}{\lambda_a(k)} y (z-1)\Big)}\frac{3.5^{3.5}}{\Gamma(3.5)}y^{\frac{5}{2}}e^{-\frac{7}{2}y}\, \mathrm{d}y}
=\frac{3.5^{3.5}}{(3.5+(1-z)\lambda_u(k)/\lambda_a(k))^{3.5}}.
\end{eqnarray}
Then, the PMF of $\mathrm{N_a}(k)$ is recovered by taking derivatives of $G$.
\end{proof}
The key assumption in our analysis is that there is one active device per resource block in each Voronoi cell. Using (\ref{NkPMF}) and $\lambda_u(k)/\lambda_a(k)=\mathbb{E}[\mathrm{N_a}(k)]$, the probability of not finding any device in the Voronoi cell of the typical aggregator at stage $k$ is 
\begin{eqnarray}
\label{0probability}
\mathbb{P}_{\mathrm{N_a}(k)}(0)=G_{\mathrm{N_a}(k)}(0)=3.5^{3.5}(3.5+\mathbb{E}[\mathrm{N_a}(k)])^{-3.5}.
\end{eqnarray}

The number of devices $\mathrm{N_a}(k)$ served by the typical aggregator is mainly determined by the fraction $\gamma$. When $\gamma$ is high, i.e., $\mathbb{E}[\mathrm{N_a}(k)]$ is low, the probability that there is no transmitting device within the Voronoi cell of the typical aggregator is not negligible. Let $p_{\mathrm{th}}(k)$ be the probability that there is at least a device in the Voronoi cell of the typical aggregator in the $k^{\rm th}$ stage. Therefore, the interference field of stage $k$ is thinned by $p_{\mathrm{th}}(k)$, and the effective density of the interfering devices at stage $k$ is $p_{\mathrm{th}}(k)\lambda_u(k)$. Using (\ref{0probability}), $p_{\mathrm{th}}(k)=1-3.5^{3.5}(3.5+\mathbb{E}[\mathrm{N_a}(k)])^{-3.5}$. From Table \ref{table:tab3}, we can see that $p_{\mathrm{th}}(k)$ is the same for $1< k < K$, and different for $k\in \{1, K\}$.

Note that in evaluating the average transmitter energy consumption in (\ref{TXenergyandcoverage}), we also require the second moment of the number of devices served per aggregator, which can be calculated adding the first and second derivatives of PGF of $\mathrm{N_a}(k)$ in (\ref{PGFofG}) evaluated at $z=1^{-}$ as follows
\begin{eqnarray}
\label{secondmoment}
\mathbb{E}[\mathrm{N_a}^2(k)]
=\frac{\lambda_u(k)}{\lambda_a(k)}+\frac{4.5}{3.5}\left(\frac{\lambda_u(k)}{\lambda_a(k)}\right)^2.
\end{eqnarray}

Conditioned on having at least one active device per resource block in each Voronoi cell, $\mathbb{E}[\mathrm{N_a}(k)\vert \mathrm{N_a}(k)>0]=p_{\mathrm{th}}(k)^{-1}\mathbb{E}[\mathrm{N_a}(k)]$ and $\mathbb{E}[\mathrm{N_a}^2(k)\vert \mathrm{N_a}(k)>0]=p_{\mathrm{th}}(k)^{-1}\mathbb{E}[\mathrm{N_a}^2(k)]$.

The uplink SIR coverage for successive stages in (\ref{uplink-SIR}) is independent of the device density and $p_{\mathrm{th}}(k)$. However, the rate coverage results for the parallel mode depend on the device density, and the interference field of stage $k$ is thinned by $p_{\mathrm{th}}(k)$, as described in Sect. \ref{ParallelStages}.

We consider two main transmission protocols, namely i) a successive transmission protocol where the stages are activated sequentially, i.e., a half-duplex sequential mode, and ii) a parallel transmission mode, which is either a full-duplex protocol where all stages are simultaneously active, or a half-duplex protocol with alternating active stages. We investigate their energy efficiencies and provide numerical comparisons for the rate-energy-delay tradeoffs in Sect. \ref{Performance}. 

\subsection{Rate Distribution for Successive Stages}
\label{SuccessiveStages}
In this mode, each transmission cycle consists of the stages operating in succession. Stages may not be repeated before a cycle is completed. This mode may provide low rate, but it has low interference since multiple stages are not active simultaneously. Let $\mathcal{K}=\{1, \hdots, K\}$ denote the set of stages and $\mathcal{R}=\{\mathrm{Rate}_1, \hdots, \mathrm{Rate}_K\}$ be the set of rates in terms of the total number of bits transmitted per unit time. Since the total number of bits that can be transmitted per $K$ units of time is $\min_{k\in\mathcal{K}}{\mathrm{Rate}_k}$, the transmission rate in successive mode is $\mathrm{R_S}=K^{-1}\min_{k\in\mathcal{K}}{\mathrm{Rate}_k}$.

\begin{remark}
\label{IndepAssum}
{\bf Dependence of hierarchical levels.}
The hierarchical levels are not independent from each other as detailed in Sect. \ref{multi-stage}, and hence, it is not tractable to analyze the joint rate distribution for successive stages. Instead, we define the rate outage as in (\ref{ratecoverageSuccessive}) where transmission rates are assumed independent. Without tracking the path of the bits (payload) transmitted, we only consider if the hierarchical transmission process is successful. Transmission from a device is successful if its payload is delivered to the BS at the end of $K$ stages. We give a comparison for the analyses with the independence assumption and the monte carlo simulations in Sect. \ref{Performance}.
\end{remark}

With the independence assumption in Remark \ref{IndepAssum}, the rate coverage for successive stages is
\begin{eqnarray}
\label{ratecoverageSuccessive}
\mathbb{P}(\mathrm{R_S}>\rho)
=\prod\limits_{k\in \mathcal{K}}{\mathbb{P}(\mathrm{Rate}_k>K\rho)}\approx\prod\limits_{k\in \mathcal{K}}{\sum\limits_{l=0}^{\infty}{\mathcal{P}_k(2^{\frac{K\rho l}{W}}-1)\mathbb{P}_{\mathrm{N_a(k)}}(l)}}.
\end{eqnarray}
If we let $P_{T_{\max}} \to \infty$, the rate coverage is simplified to
\begin{eqnarray}
\lim_{P_{T_{\max}} \to \infty}  \mathbb{P}(\mathrm{R_S}>\rho) \approx \prod\limits_{k\in \mathcal{K}}{\sum\limits_{l=0}^{\infty}{\exp{\Big(-\frac{2^{\frac{K\rho l}{W}}-1}{\frac{\alpha}{2}-1}C_{\alpha}(2^{\frac{K\rho l}{W}}-1)\Big)}\mathbb{P}_{\mathrm{N_a(k)}}(l)}}.\nonumber
\end{eqnarray}

\subsection{Rate Distribution for Full-Duplex Parallel Stages}
\label{ParallelStages}
In full-duplex parallel mode, transmissions are not interrupted during a transmission cycle unlike the successive transmission mode. All multi-hop stages operate in parallel, the $1^{\rm st}$ stage devices only transmit, and the rest of the devices both transmit and aggregate simultaneously, during all stages of the multi-hop transmission. Therefore, this mode is a full-duplex model. 

Due to the simultaneous transmissions at all levels of the hierarchical model, the interference at each stage is due to i) the interferers of that stage, i.e., the intra-stage interference, and ii) the remaining transmitting devices of the other stages, i.e., the inter-stage interference. The hierarchical levels are determined in the same manner similar to the successive mode, and also dependent in this mode, and hence, the inter-stage interference is correlated. Although full-duplex parallel mode offers high transmission rate compared to successive mode, it has higher interference since all the stages are active. The intra-stage interference in the parallel mode can be obtained in the similar manner as in the successive mode. The following lemma provides the analytical expressions for the Laplace transforms of intra-stage and inter-stage interference.

\begin{lem} \label{parallelinterferencelemma}
Given the active transmission stage $k$, the Laplace Transforms of the intra-stage interference and the inter-stage interference are given as follows:
\begin{enumerate}[(a)]
\item The Laplace transform of the intra-stage interference at stage $k$ is
\begin{multline}
\label{intrastageinterference}
\mathcal{L}_{I_{k}}(s)
\approx 
\exp\Big(- \frac{2s}{\alpha-2}\Big((1-e^{-\pi\lambda_a^{\rm eff}(k)r_c^2}(1+\pi\lambda_a^{\rm eff}(k)r_c^2))\overline{P}_TC_{\alpha}(s\overline{P}_T)\\
+(1-p_k)\pi\lambda_a^{\rm eff}(k)P_{T_{\max}}\mathbb{E}_{R_{z_k}}\left[\left.R_{z_k}^{2-\alpha}C_{\alpha}\Big(\frac{sP_{T_{\max}}}{R_{z_k}^{\alpha}}\Big)\right\vert R_{z_k}>r_c\right]\Big)\Big),
\end{multline}
where $p_k=1-\exp(-\pi\lambda_a(k)r_c^2)$ and $\lambda_{a}^{\rm eff}(k)=p_{\mathrm{th}}(k)\lambda_a(k)$.  
\item The Laplace transform of the total inter-stage interference from stages $\{\left.l \right\vert l\in\mathcal{K}, l\neq k \}$ is
\begin{multline} 
\mathcal{L}_{I_{k^c}}(s)\approx 
\prod\limits_{l\in k^c}\exp\Big(- (1-e^{-\pi\lambda_a^{\rm eff}(l)r_c^2}(1+\pi\lambda_a^{\rm eff}(l)r_c^2))(B_{\alpha}(s \overline{P}_T)+\frac{2s \overline{P}_T}{\alpha-2}C_{\alpha}(s \overline{P}_T))
\\
-(1-p_l)\pi\lambda_a^{\rm eff}(l)\mathbb{E}_{R_{z_l}}\Big[R_{z_l}^2\left.\Big(B_{\alpha}\Big(\frac{sP_{T_{\max}}}{R_{z_l}^{\alpha}}\Big)+\frac{2sP_{T_{\max}}}{(\alpha-2)R_{z_l}^{\alpha}}C_{\alpha}\Big(\frac{s P_{T_{\max}}}{R_{z_l}^{\alpha}}\Big)\Big)\right \vert R_{z_l}>r_c\Big]\Big),
\end{multline}
where $B_{\alpha}(s)={_2F_1}\Big(1,\frac{2}{\alpha},1+\frac{2}{\alpha},-\frac{1}{s}\Big)$ is obtained using the Gauss-Hypergeometric function.
\end{enumerate}
Using the Laplace transforms, the uplink SIR coverage $\mathcal{P}_k(T)$ can be found using (\ref{uplink-SIR-truncated}).
\end{lem}
\begin{proof}
See Appendix \ref{App:AppendixE}.
\qedhere
\end{proof}

\begin{cor}\label{interferencealpha4}
The Laplace transform of the total inter-stage interference for $P_{T_{\max}} \to \infty$ is 
\begin{eqnarray}
\label{interferenceinterstage}
\lim_{P_{T_{\max}} \to \infty} \mathcal{L}_{I_{k^c}}(s)
&\approx&\exp\Big(-B_{\alpha}(s \overline{P}_T)-\frac{2s \overline{P}_T}{\alpha-2}C_{\alpha}(s \overline{P}_T)\Big)^{(K-1)}.
\end{eqnarray}  
\end{cor}

\begin{lem} The uplink SIR coverage probability for the full-duplex parallel mode is given by
\begin{eqnarray}
\label{probcoverageParallel}
\mathbb{P}(\mathrm{SIR_P}>T)
\approx\prod\limits_{k\in \mathcal{K}}{\Big(p_k\mathcal{L}_{I_k}(T\overline{P}_T^{-1})\mathcal{L}_{I_{k^c}}(T\overline{P}_T^{-1})+\int_{r_c}^{\infty}{\!\mathcal{L}_{I_k}(Tr^{\alpha}P_{T_{\max}}^{-1})\mathcal{L}_{I_{k^c}}(Tr^{\alpha}P_{T_{\max}}^{-1})f_{R_k}(r)\, \mathrm{d}r}\Big)},
\end{eqnarray}
where $f_{R_k}(r)=({r}/{\sigma_k^2})e^{-r^2/2\sigma_k^2}$ for $r\geq 0$ and $\sigma_k=\sqrt{1/(2\pi\lambda_a(k))}$.
\end{lem}

\begin{proof} 
The proof is similar to that of Theorem I in \cite{Singh2014}. However, instead of one serving stage, the serving stages change sequentially. Result follows from the evaluation of (\ref{intrastageinterference}) using (\ref{uplink-SIR-truncated}). Since the total inter-stage interference in (\ref{laplaceinterstage}) is independent from the intra-stage interference, its Laplace transform is incorporated as a multiplicative term.
\end{proof}

Using (\ref{uplink-SIR}) and (\ref{interferenceinterstage}), the uplink SIR coverage for $P_{T_{\max}} \to \infty$ is given by
\begin{eqnarray}
\label{FDcoverage}
\mathbb{P}(\mathrm{SIR_P}>T)
\approx\prod\limits_{k\in \mathcal{K}}{\exp{\Big(-\frac{2T}{\alpha-2}C_{\alpha}(T)\Big)}\exp\Big(-B_{\alpha}(T)-\frac{2T}{\alpha-2}C_{\alpha}(T)\Big)^{K-1}}.
\end{eqnarray}

\begin{lem} 
\label{fullduplexratecoverage}
The rate coverage probability for the full-duplex parallel transmission mode is
\begin{eqnarray}
\label{ratecoverageParallel}
\mathbb{P}(\mathrm{R_P}>\rho)=\prod\limits_{k\in \mathcal{K}}{\mathbb{P}(\mathrm{Rate}_k>\rho)}
=\prod\limits_{k\in \mathcal{K}}{\sum\limits_{l=0}^{\infty}{\mathbb{P}(\mathrm{SIR}_k>2^{\frac{\rho l}{W}}-1|\mathrm{N_a}(k)=l)\mathbb{P}_{\mathrm{N_a(k)}}(l)}}.
\end{eqnarray}
\end{lem}

\begin{proof}
Proof follows from the combination of (\ref{probcoverageParallel}) with the definition of rate. 
\end{proof}
The rate coverage expression for the full-duplex parallel mode is validated in Sect. \ref{Performance}. 

\begin{remark}
Although the full-duplex parallel transmission strategy is probably not feasible for M2M communication, it is helpful to have a comparison of the rate-energy tradeoffs of both models, and provided for completeness. The models described above, i.e., the sequential mode which is half-duplex by design, and the full-duplex parallel mode, are the two principle design schemes that mainly differ in terms of their total rate coverages and energy consumptions. 
\end{remark}
Next, we introduce the half-duplex parallel transmission, which is a feasible scheme for M2M.

\subsection{Rate Distribution for Half-Duplex Parallel Stages}
\label{HalfDuplexParallelStages}
The full-duplex parallel mode can be transformed into a half-duplex communication scheme. In this mode, at a particular time slot, only the even or odd stages are active. Analysis of this model is quite similar to the parallel mode analysis, using only the active stages when calculating the inter-stage interference. The SIR coverage of the half-duplex mode can be characterized in a fashion similar to full-duplex mode SIR coverage in (\ref{probcoverageParallel}), and is given by the following lemma. 
\begin{lem} The uplink SIR coverage probability for the half-duplex parallel mode is given by
\begin{eqnarray}
\label{probcoverageHalfDuplexParallel}
\mathbb{P}(\mathrm{SIR_P}>T)
\approx\prod\limits_{k\in \mathcal{K}_{\rm odd}}{\Big(p_k\mathcal{L}_{I_k}(T\overline{P}_T^{-1})\mathcal{L}_{I_{k^c}}(T\overline{P}_T^{-1})+\int_{r_c}^{\infty}{\!\mathcal{L}_{I_k}(Tr^{\alpha}P_{T_{\max}}^{-1})\mathcal{L}_{I_{k^c}}(Tr^{\alpha}P_{T_{\max}}^{-1})f_{R_k}(r)\, \mathrm{d}r}}\Big).
\end{eqnarray}
where calculations of $\lambda_a^{\rm eff}(k)$ and the distributions of $\{R_{z_k}\}$ in (\ref{intrastageinterference}) and (\ref{laplaceinterstage}) are done over the set of active stages and $\mathcal{K}_{\rm odd}$ denotes the set of active stages, i.e., the odd stages.

The uplink SIR coverage for $P_{T_{\max}} \to \infty$ is given by
\begin{eqnarray}
\label{HDcoverage}
\mathbb{P}(\mathrm{SIR_P}>T)
\approx \prod\limits_{k\in \mathcal{K}_{\rm odd}}{\exp{\Big(-\frac{2T}{\alpha-2}C_{\alpha}(T)\Big)}\exp\Big(-B_{\alpha}(T)-\frac{2T}{\alpha-2}C_{\alpha}(T)\Big)^{\frac{K}{2}-1}}.
\end{eqnarray}
\end{lem}
For the half-duplex mode, the rate coverage expression in (\ref{ratecoverageParallel}) changes in accordance with (\ref{probcoverageHalfDuplexParallel}), which can be found by following the steps in Lemma \ref{fullduplexratecoverage}. Since only half of the stages are active simultaneously, the overall rate is half the rate of the active stages. Thus, to achieve a rate threshold of $\rho$, the rate threshold for the active stages should be set to $2\rho$. Note also that $p_{\mathrm{th}}(k)$'s are modified in accordance with the updated values of $\mathbb{E}[\mathrm{N_a}(k)]$.

\subsection{Expected Communication Delay}
\label{energy-delay}
Average communication delay depends on the number of hops in the transmission and the communication rate. Using the bounds (\ref{K-upper}) and (\ref{K-lower}) on the total number of hops $K$ in Sect. \ref{number-stages}, and denoting the delay of the typical device due to the transmission of a payload of $M$ bits at rate $R$ over $K$ stages by $\mathcal{D}(R,K)$, it satisfies $\mathcal{D}(R,K)=t(R,K)-t(R,1)$, where $t(R,1)$ is the direct transmission duration as a function of the transmission rate $R$ in the unit of bits/s and the payload $M$ in the unit of bits, and it is given as $t(R,1)=MR^{-1}$, and $t(R,K)$ is the transmission duration over $K$ hops. We calculate the average transmission delay conditioned on the devices that satisfy the minimum rate threshold $\rho$ over $K$ hops, which is obtained as
\begin{eqnarray}
\label{delayformula}
\mathbb{E}[\left.\mathcal{D}(R,K)\right\vert R>\rho]=\mathbb{E}\left[\left.t(R,K)\right\vert R>\rho\right]-M\mathbb{E}\Big[\left.\frac{1}{R}\right\vert R>\rho\Big],
\end{eqnarray}
where $\mathbb{E}\left[\left.\frac{1}{R}\right\vert R>\rho\right]=\int\nolimits_{0}^{1/\rho}{ \! \Big[1-\mathbb{P}\big(R\geq \frac{1}{t}\big)\Big]\, \mathrm{d}t}\approx\sum\limits_{l=1}^{\infty}{\int\nolimits_{0}^{1/\rho}{ \! \Big[1-\mathcal{P}(2^{\frac{l}{Wt}}-1)\Big] \, \mathrm{d}t}\mathbb{P}_{\mathrm{N_a}}(l)}$, which can be calculated using the rate distribution given in (\ref{coverage-prob}).

The expected delay depends on the transmission protocol. As discussed in Sect. \ref{rate}, we mainly focus on two different transmission protocols, i.e., successive and full-duplex parallel transmission modes as described in detail in Sects. \ref{SuccessiveStages} and \ref{ParallelStages}, respectively.

\begin{figure}[t!]
\centering
\begin{minipage}{.47\textwidth}
  \centering
  \includegraphics[width=\linewidth]{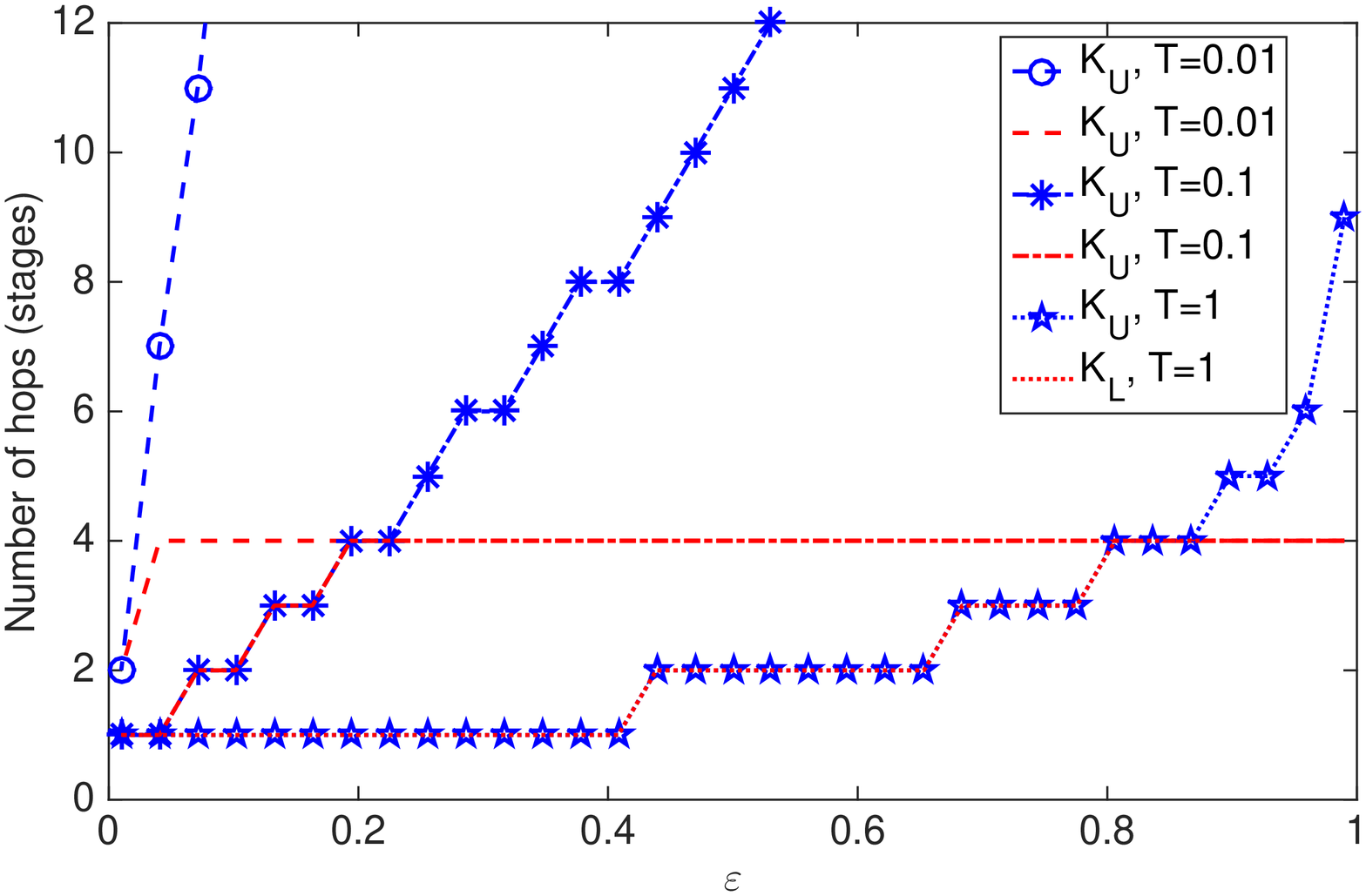}
  \caption{\footnotesize{Number of hops versus the total SIR outage probability constraint, $\varepsilon$ for variable $T$.}}
  \label{fig-Pout2}
\end{minipage}
\hfill
\begin{minipage}{.47\textwidth}
  \centering
  \includegraphics[width=\linewidth]{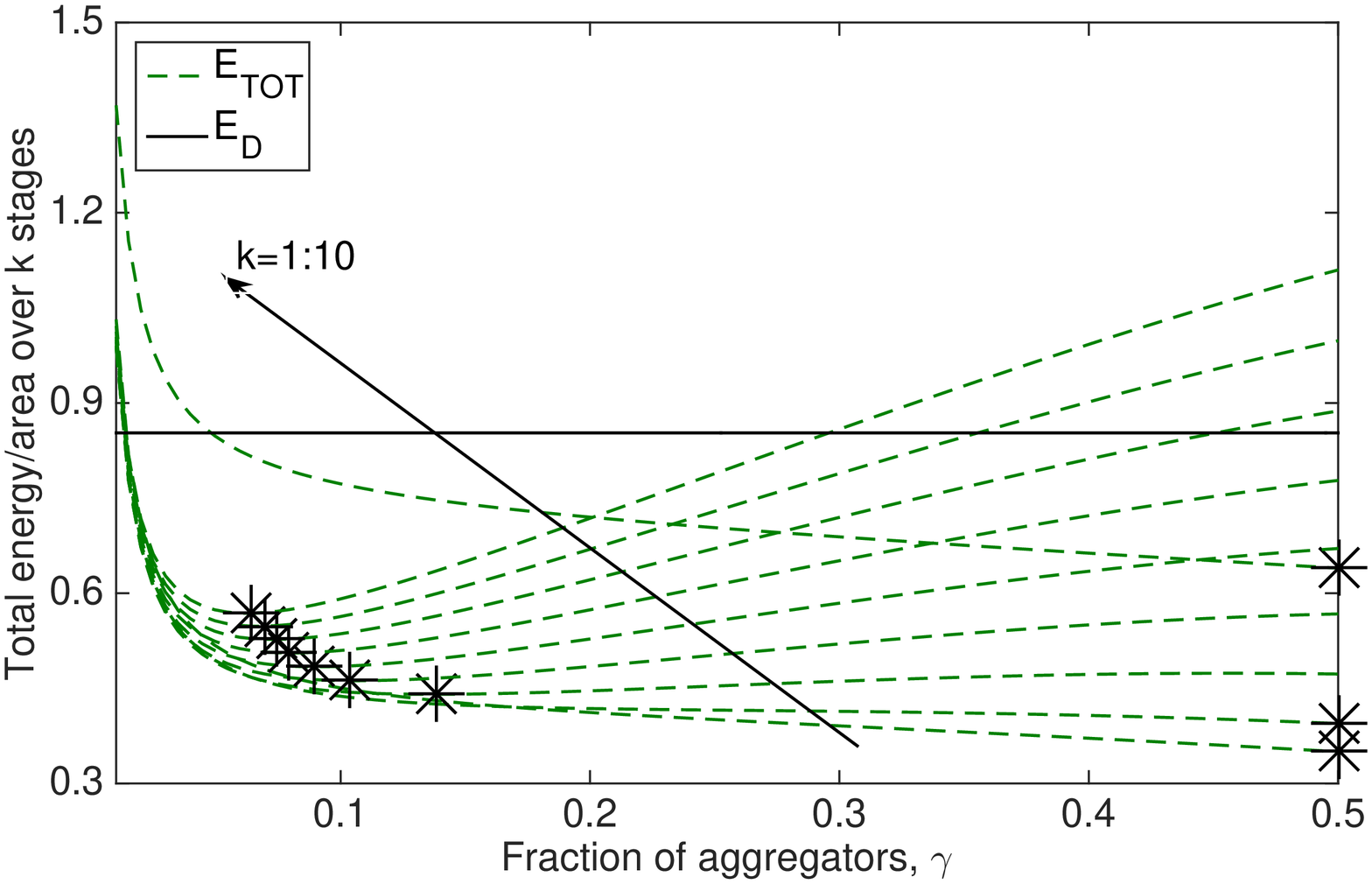}
  \caption{\footnotesize{Total energy density, with optimal fraction of aggregators, $\gamma_{\textrm{opt}}$, marked for $P_{\rm LO}=5{\rm mW}$ \cite{Wang2006}.}}
  \label{fig-3}
\end{minipage}
\end{figure}

{\bf Successive Transmission.} 
Using the distribution of the sequential mode transmission rate $R_S$ given in (\ref{ratecoverageSuccessive}), the expected communication duration for $K$ stages can be obtained as
\begin{eqnarray}
\label{sequentialduration}
\mathbb{E}[\left.t(R_S,K)\right\vert R_S>\rho]
\approx MK\int\nolimits_{0}^{1/\rho} \! \Big[1-\prod\limits_{k\in \mathcal{K}}\sum\nolimits_{l=0}^{\infty}\mathcal{P}_k(2^{\frac{Kl}{Wt}}-1) \mathbb{P}_{\mathrm{N_a(k)}}(l)\Big] \, \mathrm{d}t.
\end{eqnarray}

{\bf Parallel Transmission.} 
The distribution of the full-duplex parallel mode transmission rate $R_P$ is given in (\ref{ratecoverageParallel}). Therefore, the expected communication duration is obtained as
\begin{eqnarray}
\label{parallelduration}
\mathbb{E}[\left.t(R_P,K)\right\vert R_P>\rho]
\approx MK\int\nolimits_{0}^{1/\rho} \! \Big[1-\prod\limits_{k\in \mathcal{K}}\sum\limits_{l=0}^{\infty}{\mathbb{P}(\mathrm{SIR}_k>2^{\frac{l}{Wt}}-1)\mathbb{P}_{\mathrm{N_a(k)}}(l)} \Big]\, \mathrm{d}t,
\end{eqnarray}
where $\mathbb{P}(\mathrm{SIR}_k>2^{\frac{l}{Wt}}-1)$ can be calculated using (\ref{probcoverageParallel}).

The expected delay for the sequential (or full-duplex parallel) mode can be found by combining (\ref{sequentialduration}) (or (\ref{parallelduration})) with the delay expression in (\ref{delayformula}). 

%%%%%%%%%%%%%%%%%%%%%%%%%%%%%%%%%%%%%%%%%%%%%%%%%%%%%%%%%%%%%%%%%%%%%%%%%%%%%%%%%%%%%%%%%%%%%%%%%%%%%%%%%%%%%%%%%%%%%%%%%%
\section{Numerical Results}
\label{Performance}
For the performance evaluation of the aggregator-based M2M communication model in Sects. \ref{single-stage}-\ref{multi-stage} and the rate coverage models in Sect. \ref{rate}, we incorporate the power-dissipation factors of IEEE 802.11b (WiFi) technology \cite{Wang2006}. The simulation parameters are shown in Table \ref{table:tab2}.
 
{\bf Optimal number of hops.} We illustrate the trend between the number of required stages versus the outage probability constraint for different thresholds $T=\{0.01, 0.1, 1\}$ in Fig. \ref{fig-Pout2}. For low outage probabilities, i.e., small $\varepsilon$ values, $K_{\mathrm U}$ should be also low as the average transmission range is long. As $\varepsilon$ increases, the maximum transmission range decreases and $K_{\mathrm U}$ increases. The variation of the lower bound on the number of hops, $K_{\mathrm L}$, for varying $\varepsilon$ and $K_{\mathrm U}$ is also shown.

{\bf Optimal fraction of receivers decay with the number of stages.} We evaluate the performance of the energy model in Sects. \ref{single-stage} and \ref{multi-stage},  for $K\in\{1,\hdots, 10\}$. In Figs. \ref{fig-3}-\ref{fig-5},  we observe the variation of the average total energy density with respect to $\gamma$ for different $P_{\rm LO}$, where $E_D$ (solid line) corresponds to the energy density of direct transmission to the BS. The optimal values of $\gamma$ that minimize the total energy density, i.e., $\gamma_{\textrm{opt}}$ (marked), are decreasing with $K$.  

\begin{figure*}[t!]
\centering
\begin{minipage}{.47\textwidth}
  \centering
  \includegraphics[width=\linewidth]{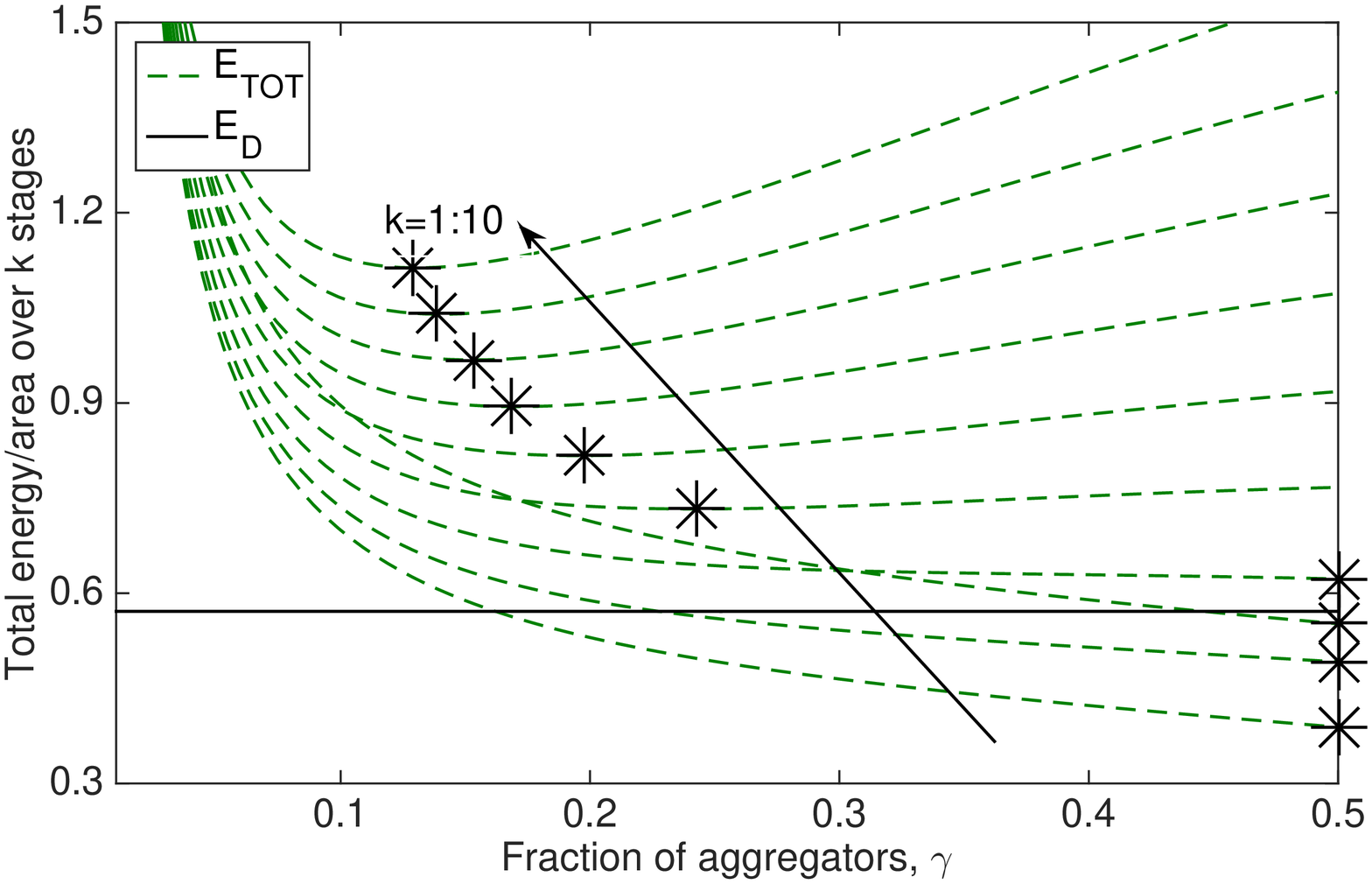}%{ER200highPLOtotal.pdf}
  \caption{\footnotesize{Total energy density, with optimal fraction of aggregators, $\gamma_{\textrm{opt}}$, marked for $P_{\rm LO}=P_{\rm RX}/4$ \cite{Wang2006}.}}
  \label{fig-5}
\end{minipage}
\hfill
\begin{minipage}{.47\textwidth}
  \centering
  \includegraphics[width=\linewidth]{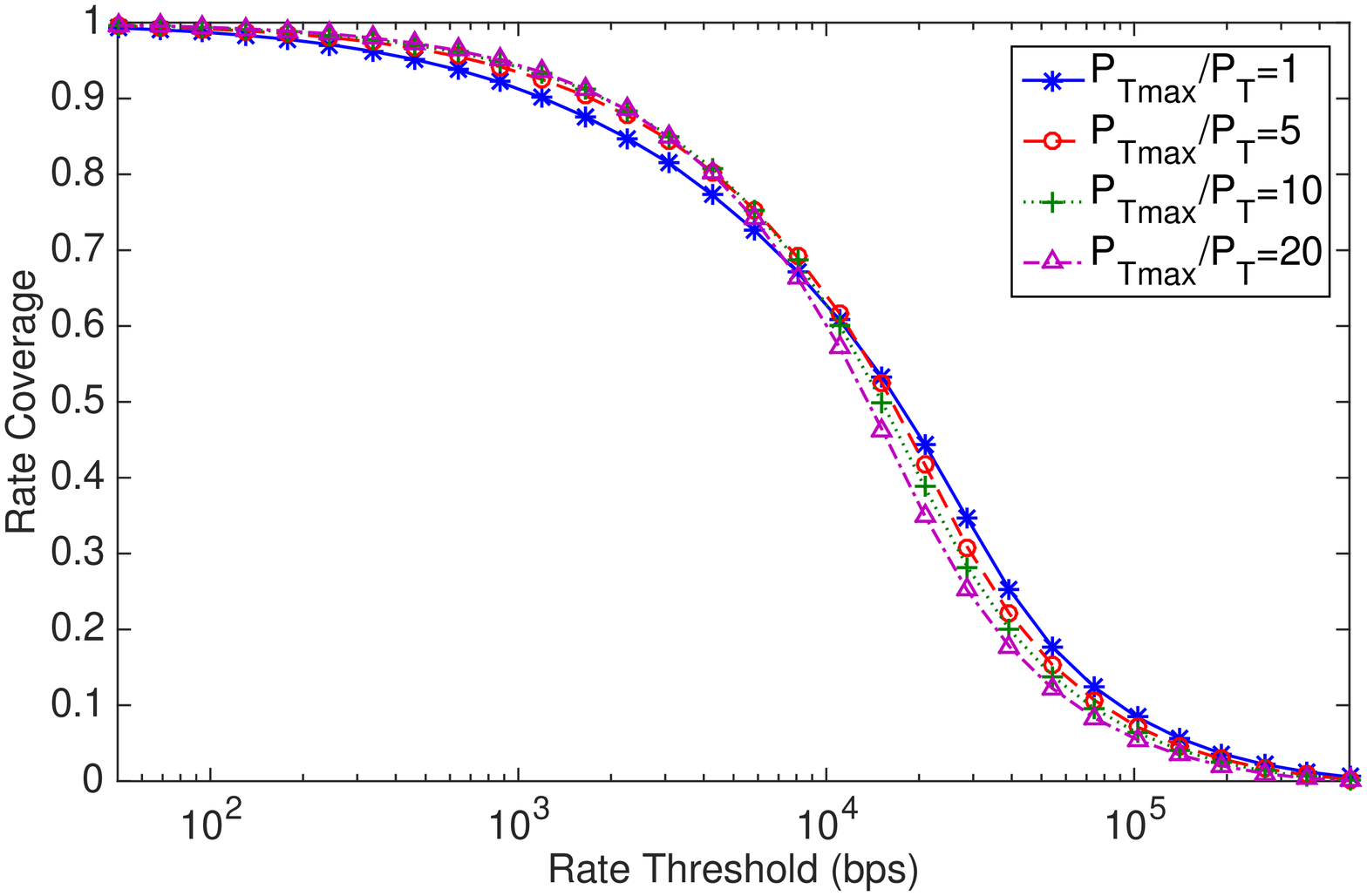}
  \caption{\footnotesize{Rate coverage distribution with maximum power constraint $P_{T_{\max}}/\overline{P}_T=[1,5,10,20]$.}}
  \label{fig-RatePowerConstraint}
\end{minipage}
\end{figure*}

\begin{table}[b!]\footnotesize%\scriptsize
\begin{center}
  \begin{tabular}{|l|l||l|l|}
    %\begin{tabular}{|l|l|}
   % \hline
   % \bf Parameter & \bf Value/Range \\ 
   \hline
    \bf Parameter & \bf Value/Range & \bf Parameter & \bf Value/Range \\ 
    \hline
    Total bandwidth & $W=10^5\mathrm{Hz}$ & Power consumption of RX (TX) & $P_{\rm RX}=200$ ($P_{\rm TX}=100$) {\rm mW}\\
    \hline
     Path loss exponent & $\alpha=4$ & Network size & $2R \times 2R$ sq. $\mathrm{km}$, $R=2.5 \mathrm{km}$\\
    \hline
    Device (BS) density & $\lambda=10^{3}$ ($\lambda_{\mathrm{BS}}=1$) per sq. km & Payload & $M=100$ bits\\
    \hline
%   Small-scale fading $g$ & $\exp{(1)}$  & %Large-scale fading $S$ & Lognormal ($8$ dB std. dev.) 
%    & \\
%    \hline   
  \end{tabular} 
\end{center}
\caption{\footnotesize{Simulation parameters.}}
\label{table:tab2}
\end{table}

The simulation setup for the verification of analytical rate models developed in Sect. \ref{rate} is as follows. Device locations are distributed as PPP over a square region of size $5 \times 5$ sq. $\mathrm{km}$. Note that the density of aggregators at stage $k$ is $\lambda \gamma^k$, i.e., the number of aggregators decay geometrically. Therefore, for high $K$ values we need a region with much larger area for the validation of the model, but scaling the region increases the computational complexity exponentially. Therefore, we restrict ourselves to $K=1:3$, which is also consistent with the plot characteristics on the number of hops illustrated in Fig. \ref{fig-Pout2}, and investigate the performance of the proposed model. Furthermore, we approximate (\ref{ratecoverageSuccessive}) by truncating $\mathrm{N_a}$ over the range $l=0, 1, \hdots, L$, where $L$ is set to 20. This is a reasonable range as the fraction of aggregators is chosen to be $\gamma=0.1$, i.e., there are 9 devices per aggregator on average.

{\bf Rate coverage for half-duplex and full-duplex modes.} The rate coverage for a single-stage model for different maximum power levels is shown in Fig. \ref{fig-RatePowerConstraint}. As $P_{T_{\max}}/\overline{P}_T$ increases, coverage is improved for thresholds less than $10^4$ bps. For the rest of the simulations, we assume $P_{T_{\max}}\to\infty$ for tractability. The rate coverage for the sequential and the full-duplex modes for different number of hops are illustrated in Figs. \ref{fig-sequentialrates} and \ref{fig-parallelrates}. The analytical results do not exactly match the simulations, and the difference is due to the {\em interstage independence assumption} in Sect. \ref{number-stages}, and the {\em independent power control assumption} in \cite{Novlan2013}, which is the main limitation of this work. Comparing Figs. \ref{fig-sequentialrates} and \ref{fig-parallelrates}, we observe that the full-duplex mode does not offer higher rate coverage compared to sequential mode, which is due to {\em inter-stage interference} as detailed in Sect. \ref{ParallelStages}. Thus, in terms of rate coverage, sequential mode is preferable over full-duplex mode. From (\ref{HDcoverage}) and (\ref{FDcoverage}), the half-duplex parallel mode has higher coverage than full-duplex mode for the same $K$. Due to limited space, we do not provide an illustration for that scheme.

\begin{figure}[t!]
\centering
\begin{minipage}{.47\textwidth}
  \centering
  \includegraphics[width=\linewidth]{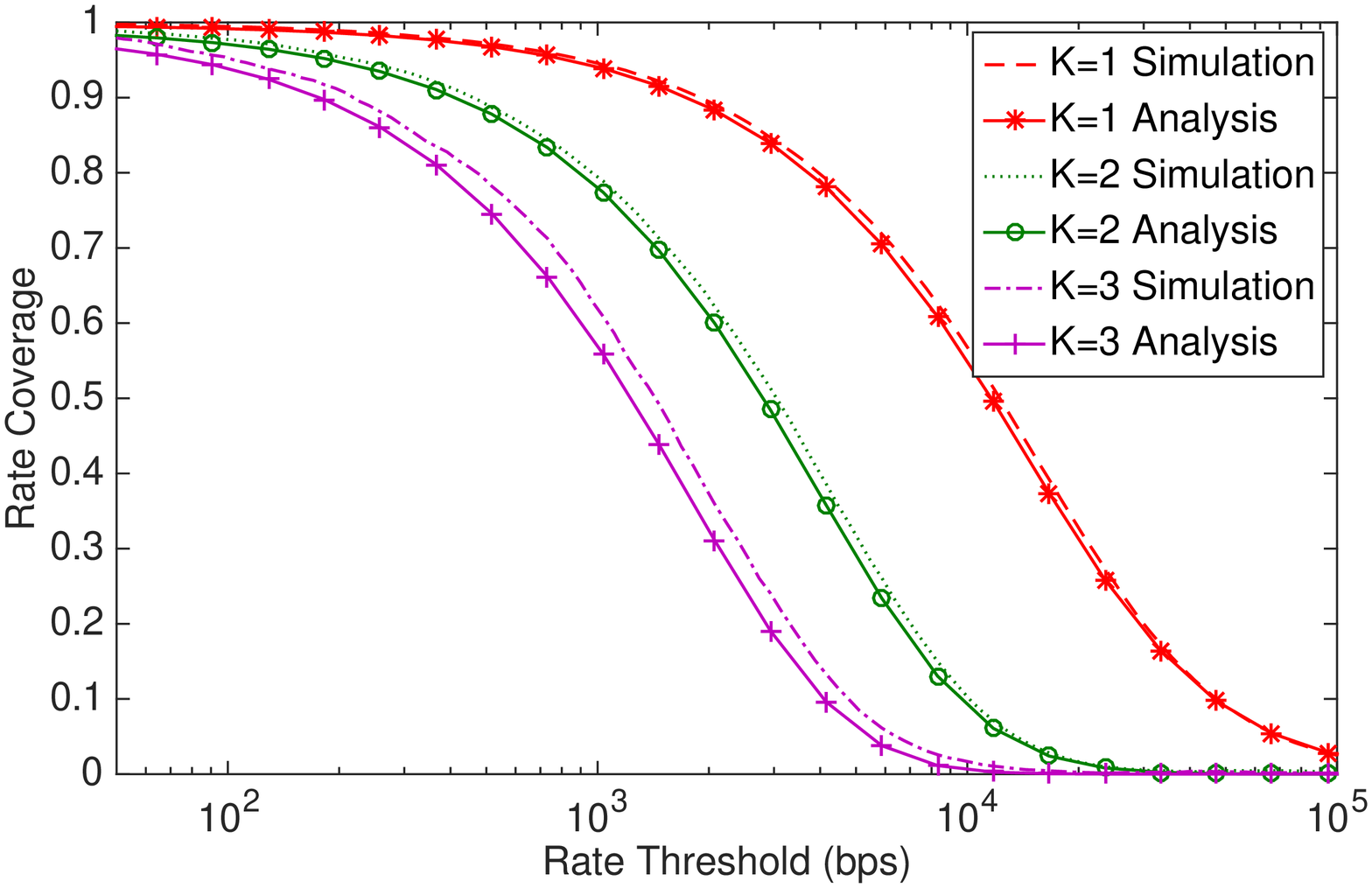}%{sequentialratesupdated2.pdf}
  \caption{\footnotesize{Rate coverage distributions in sequential mode for $K=\{1, 2, 3\}$ stages.}}
  \label{fig-sequentialrates}
\end{minipage}
\hfill
\begin{minipage}{.47\textwidth}
  \centering
  \includegraphics[width=\linewidth]{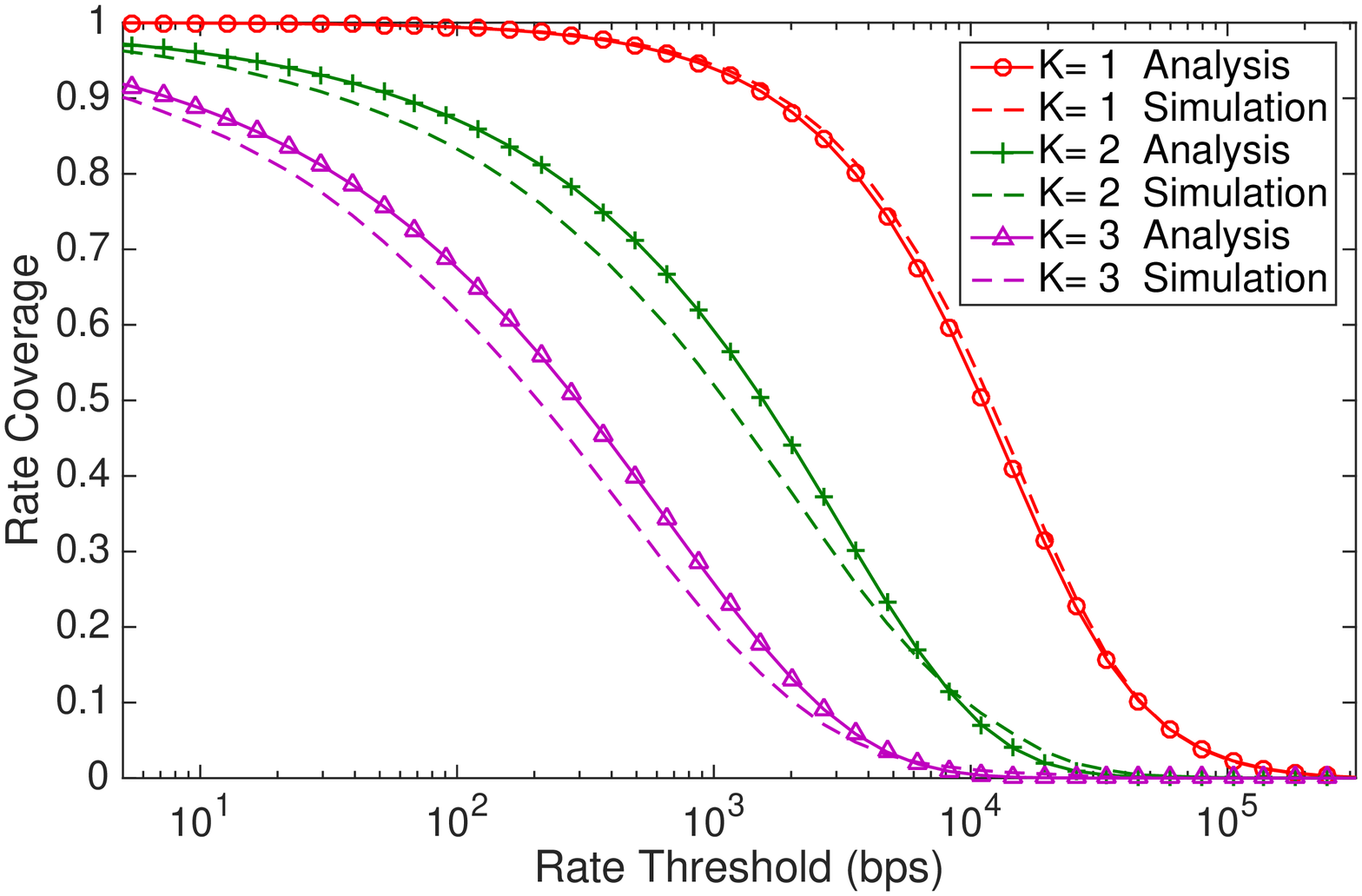}
  \caption{\footnotesize{Rate coverage distributions in full-duplex parallel mode for $K=\{1, 2, 3, 4\}$ stages.}}
  \label{fig-parallelrates}
\end{minipage}
\end{figure}

{\bf Rate-energy tradeoffs.} We compare the average total energy density for the sequential and full-duplex parallel schemes in Fig. \ref{fig-Rate-Energy-Tradeoff1}. The rate coverage decays as the hop number increases, which reduces the chance of successful transmissions, and the overall energy consumption. To make a relevant comparison based on the energy consumptions of sequential and parallel modes, we consider multiple transmission cycles equal to the number of stages, i.e., $K=2$. On the plot legend, $\mathcal{E}_P$ and $\mathcal{E}_S$ stand for the total energy densities calculated based on the actual rate coverage probabilities, and $Q_P$ and $Q_S$ stand for the rate outage probabilities, i.e., complements of the rate coverage probabilities, of the full-duplex parallel and sequential transmission modes, respectively. In full-duplex mode, all devices are always active and transmitting, but the communication rate is low due to higher interference, hence, its energy consumption is less than the successive mode. We also investigate the average total energy density for the half-duplex parallel scheme with $K=2$ in Fig. \ref{fig-Rate-Energy-Tradeoff2}. In the figure legend, HD stands for half-duplex. This scheme has higher coverage probability as only half of the stages are operating simultaneously, and higher transmit distance as the aggregator to device fraction is $\gamma^2\ll \gamma$, which yields higher energy consumption. 

Different transmission modes have revealed the tradeoffs between the coverage and the energy requirements. Full-duplex mode has low energy density, but low rate coverage and high delay. Sequential mode also has low energy density, high rate coverage and low delay. On the other hand, half-duplex parallel mode has higher energy density ($\times 2$) and higher rate coverage compared to full-duplex mode. Unlike the full-duplex mode, which is not convenient for M2M communication despite being energy efficient, the half-duplex parallel mode is favorable with higher energy consumption. However, its energy consumption can be reduced by readjusting the aggregator fractions for alternating stages. Since M2M is delay tolerant, sequential mode is also feasible and a low-cost technique. Considering the operating regime for M2M devices and the simplicity of their design, sequential mode is preferable as it has low energy density and high coverage.

\begin{figure}[t!]
\centering
\begin{minipage}{.47\textwidth}
  \centering
  \includegraphics[width=\linewidth]{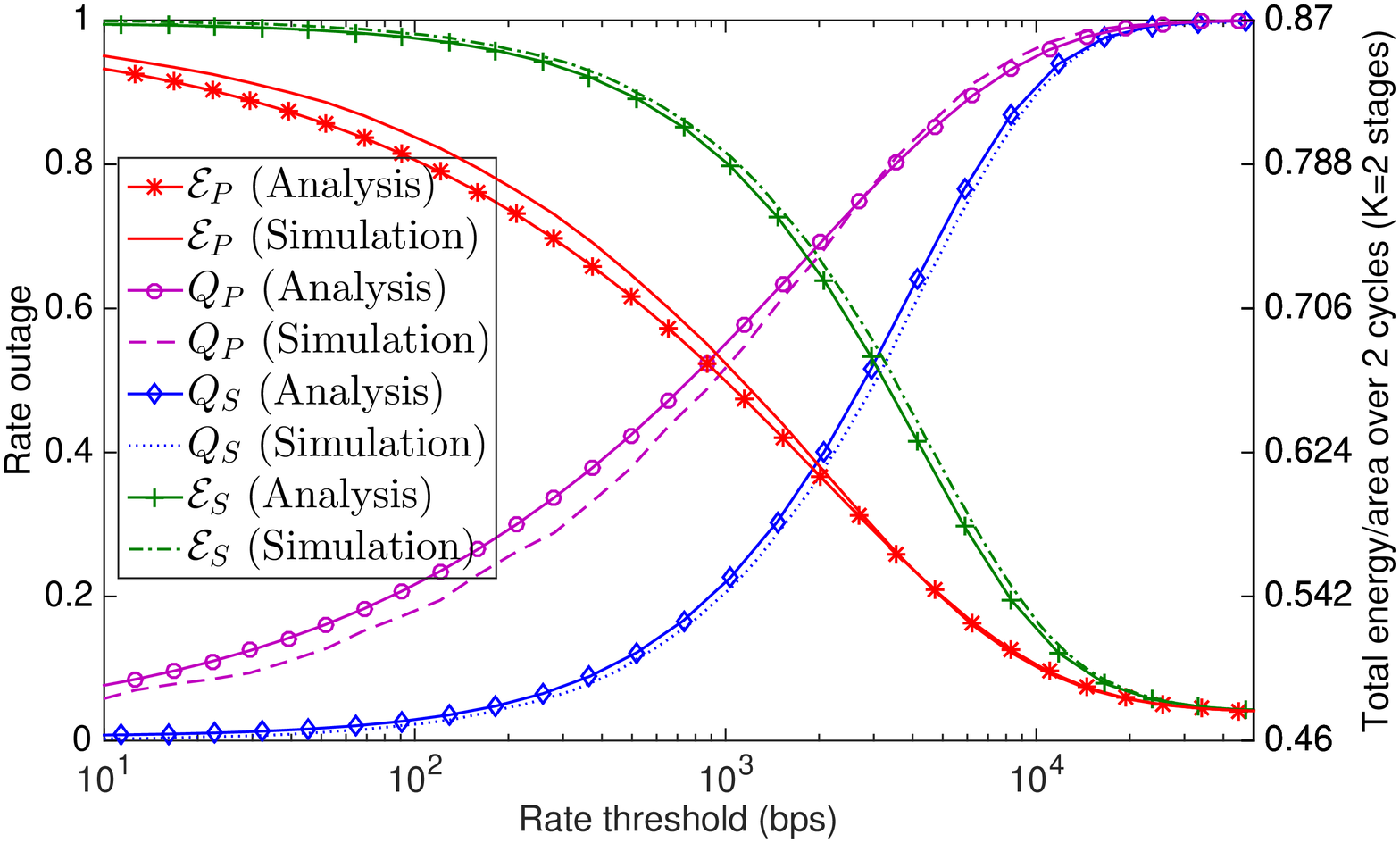}%{Rate-Energy-Tradeoff1.pdf}
  \caption{\footnotesize{A rate outage-total energy density tradeoff for sequential and parallel modes, $K=2$.}}
  \label{fig-Rate-Energy-Tradeoff1}
\end{minipage}
\hfill
\begin{minipage}{.47\textwidth}
  \centering
  \includegraphics[width=\linewidth]{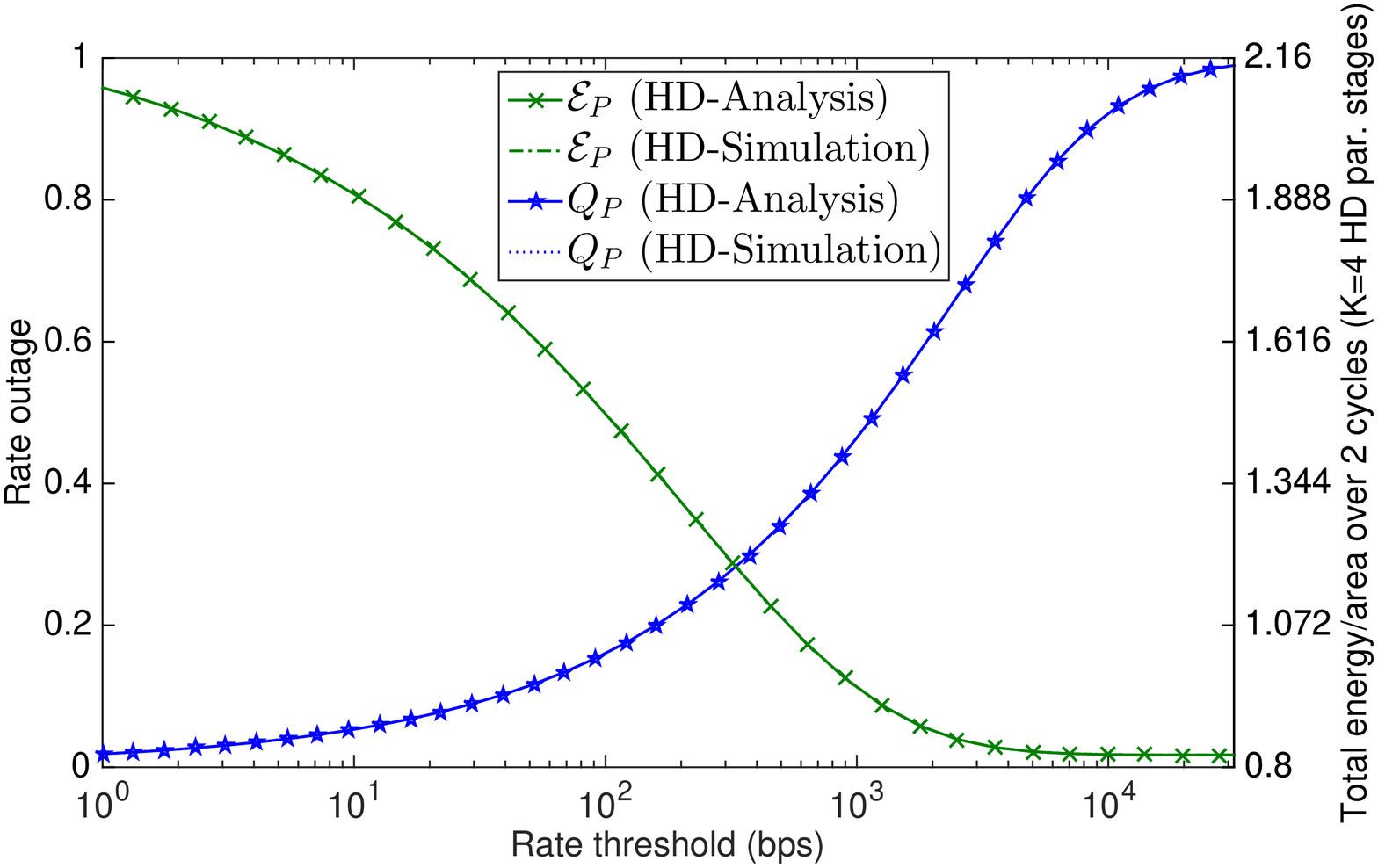}%{Rate-Energy-Tradeoff2.pdf}
  \caption{\footnotesize{A rate outage-total energy density tradeoff comparison of half-duplex parallel model for $K=2$.}}
  \label{fig-Rate-Energy-Tradeoff2}
\end{minipage}
\end{figure}

%%%%%%%%%%%%%%%%%%%%%%%%%%%%%%%%%%%%%%%%%%%%%%%%%%%%%%%%%%%%%%%%%%%%%%%%%%%%%%%%%%%%%%%%%%%%%%%%%%%%%%%%%%%%%%%%%%%%%%%%%%
\section{Conclusions}
We study a general multi-hop-based uplink communication scheme for M2M communication, and develop a new data aggregation model for M2M devices, using tools from stochastic geometry. To the best of our knowledge, this is the first work on power-limited communication providing a unified energy consumption model with coverage in cellular networks. Our results show that the uplink coverage characteristics dominate the trend of the energy consumption for the proposed transmission modes. Considering the operating regime of M2M devices, sequential and half-duplex parallel modes are more feasible compared to full-duplex mode.

Interesting extensions would include the minimization of the energy expenditure through joint optimization of the optimal number of multi-hop stages and fraction of aggregators. Better energy models can be developed to incorporate the different states of the transceiver, i.e., on, idle, sleep and off states, which will provide a more accurate energy model for the proposed communication scheme. Strategies for synchronization of transmissions is also important to prevent multi-hop delays and save receiver energy consumption. Multi-slope path loss models \cite{Zhang2015} would more accurately model the effect of interference and noise for networks with increasing device density.

%%%%%%%%%%%%%%%%%%%%%%%%%%%%%%%%%%%%%%%%%%%%%%%%%%%%%%%%%%%%%%%%%%%%%%%%%%%%%%%%%%%%%%%%%%%%%%%%%%%%%%%%%%%%%%%%%%%%%%%%%%

%\section*{Acknowledgment}
%\addcontentsline{toc}{section}{Acknowledgment}
%The authors thank Sarabjot Singh for helpful feedback.

%%%%%%%%%%%%%%%%%%%%%%%%%%%%%%%%%%%%%%%%%%%%%%%%%%%%%%%%%%%%%%%%%%%%%%%%%%%%%%%%%%%%%%%%%%%%%%%%%%%%%%%%%%%%%%%%%%%%%%%%%%
\begin{appendix}
\section{Appendices}
\subsection{Proof of Theorem \ref{maintheo}} \label{App:AppendixA}
Let $\Psi_a=\sum_{n}\delta_{X_n}$, $\Psi_u=\sum_{n}\delta_{\bar{X}_n}$, and $V_0$ be the Voronoi cell of the typical aggregator $X_0\in \Psi_a$ located at the origin. Notice that $x\in V_0$ if and only if $\Psi_a(\mathrm{B}^o(x,\norm{x}))=0$, where $\mathrm{B}^o(x,\norm{x})$ is the open ball centered at $x$ and of radius $\norm{x}$. Hence,
\begin{eqnarray}
\mathbb{E}_{\Psi_a}^0 \int_{\mathbb{R}^2}{ \! f(x)1_{x\in V_0} \, \mathrm{d}x}
\!=\!\int_{\mathbb{R}^2}{ \! f(x)\mathbb{P}_{\Psi_a}^0\left[\Psi_a(\mathrm{B}^o(x,\norm{x}))=0\right] \, \mathrm{d}x},\nonumber
\end{eqnarray}
which is the expectation with respect to the conditional (Palm) probability conditioned on $0\in \Psi_a$. As $\Psi_a$ is PPP due to thinning, from Slivnyak's theorem \cite{BaccelliBook1}, the RHS can be rewritten as
\begin{eqnarray}
\int_{\mathbb{R}^2}{ \! f(x)\mathbb{P}_{\Psi_a}^0\left[\Psi_a(\mathrm{B}^o(x,\norm{x}))=0\right] \, \mathrm{d}x}=\int_{\mathbb{R}^2}{ \! f(x)\mathbb{P}\left[\Psi_a(\mathrm{B}^o(x,\norm{x}))=0\right] \, \mathrm{d}x}.\nonumber
\end{eqnarray}
Using the Poisson law, we have $\mathbb{E}_{\Psi_a}^0{\sum\limits_{n}{f(\bar{X}_n)1_{\bar{X}_n\in V_0}}}=\lambda_u \int_{\mathbb{R}^2}{ \! f(x)e^{-\lambda_a \pi \norm{x}^2} \, \mathrm{d}x}$. We let $f(x)=\min\{P_{T_{\max}}, \overline{P}_T\norm{x}^{\alpha}\}$ be the transmit power with a maximum power constraint $P_{T_{\max}}$. Then, we have the following final result for the mean total uplink power of the devices in the Voronoi cell of the typical aggregator, i.e., the total power consumption of the PAs: 
\begin{eqnarray}
 P(\lambda_a)=\frac{2\pi\lambda_u \overline{P}_T}{\eta} \int_{0}^{r_c}{ \! r^{\alpha}e^{-\lambda_a \pi r^2} r\, \mathrm{d}r}+\frac{2\pi\lambda_u P_{T_{\max}}}{\eta}\int_{r_c}^{\infty}{ \! e^{-\lambda_a \pi r^2} r\, \mathrm{d}r},\nonumber
\end{eqnarray}
where $r_c=\big({P_{T_{\max}}}/{\overline{P}_T}\big)^{1/\alpha}$, and the final result follows from employing a change of variables $u=\lambda_a\pi r^2$, and $\gamma(s,x)=\int_{0}^{x}{ \! t^{s-1}e^{- t} \, \mathrm{d}t}$ is the lower incomplete gamma function. 

\subsection{Proof of Lemma \ref{meanPR}}\label{App:AppendixB}
	The average received power from a device at distance $x$ from the typical aggregator is $P_R(x)=\min\{P_{T_{\max}}\norm{x}^{-\alpha},\overline{P}_T\}$. The average number of devices within distance $r_c=\big( {P_{T_{\max}}}/{\overline{P}_T}\big)^{1/\alpha}$ of the typical aggregator is $\mathbb{E}\left[\mathrm{N_a}(r_c)\right]=\frac{\lambda_u}{\lambda_a}\left(1-\exp\left(-\pi\lambda_a r_c^2\right)\right)$. Using the Poisson law, the mean total received power is given by
	\begin{eqnarray}
	\overline{P}_R=\mathbb{E}_{\Psi_a}^0{\sum\limits_{n}{P_R(\bar{X}_n)1_{\bar{X}_n\in V_0}}}=\lambda_u \int_{\mathbb{R}^2}{ \! P_R(x)e^{-\lambda_a \pi \norm{x}^2} \, \mathrm{d}x}=
	2\pi\lambda_u\int_{0}^{\infty}{ \! P_R(r) e^{-\lambda_a \pi r^2} r\, \mathrm{d}r},\nonumber
	\end{eqnarray}
	where the final result follows from incorporating $P_R(x)$ to split the integral into two parts, employing a change of variables $\lambda_a\pi r^2=u$ and the definition of incomplete Gamma function.

\subsection{Proof of Proposition \ref{mainprop}} \label{App:AppendixC}
Note that $\mathcal{E}(\pmb{\lambda_a}(K))=\sum_{k=1}^{K}{c_k(\gamma)}$, where $c_k(\gamma)$ is given in Proposition \ref{simple}. Note also that $\min\{\mathcal{E}(\pmb{\lambda_a}(K))\}+\min\{c_{K+1}(\gamma)\}\leq \min\{\mathcal{E}(\pmb{\lambda_a}(K+1))\}$. Adopting $\gamma^{(K)}$ to represent the ratio for $K$ stages, and letting $\gamma^{(K)*}$ be the optimal value for $K$ stages, optimal solutions of $\mathcal{E}(\pmb{\lambda_a}(K))$ and $\mathcal{E}(\pmb{\lambda_a}(K+1))$ are $\gamma^{(K)*} = \argmin \mathcal{E}(\pmb{\lambda_a}(K))$ and $\gamma^{(K+1)*} = \argmin \mathcal{E}(\pmb{\lambda_a}(K+1))$, yielding
\begin{eqnarray}
\label{gammabounds}
\sum\nolimits_{k=1}^{K}{c_k(\gamma^{(K)*})}+\min\{c_{K+1}(\gamma)\}\leq\sum\nolimits_{k=1}^{K+1}{c_k(\gamma^{(K+1)*})}.
\end{eqnarray}
It is also clear that
\begin{eqnarray}
\label{gammasum}
\sum\nolimits_{k=1}^{K}{c_k(\gamma^{(K)*})}\leq\sum\nolimits_{k=1}^{K}{c_k(\gamma^{(K+1)*})}.
\end{eqnarray}
Using (\ref{gammabounds}) and (\ref{gammasum}), we have $\min\{c_{K+1}(\gamma)\}\leq {c_{K+1}(\gamma^{(K+1)*})}$. Assume that $\gamma^{(K+1)*}> \gamma^{(K)*}$. Then, from Remark \ref{rem}, $c_{k}(\gamma^{(K)*})>c_{k}(\gamma^{(K+1)*})$ for $k\leq K-1$, and $c_{K}(\gamma^{(K)*})<c_{K}(\gamma^{(K+1)*})$. Using (\ref{gammasum}), and from $\gamma^{(K+1)*}> \gamma^{(K)*}$, we can infer that $c_{K+1}(\gamma^{(K+1)*})$ is increasing in $\gamma^{(K+1)*}$, and reduce $c_{K+1}(\gamma^{(K+1)*})$, i.e., the energy density at stage $K+1$, by decreasing $\gamma^{(K+1)*}$ to $\gamma^{(K)*}$, which implies $\gamma^{(K+1)*} \neq \argmin \mathcal{E}(\pmb{\lambda_a}(K+1))$, giving a contradiction. Thus, $\gamma^{(K+1)*}\leq \gamma^{(K)*}$.

\subsection{Proof of Lemma \ref{maincoveragelemma}} \label{App:AppendixD}
The Laplace transform of the interference can be written as	
\begin{eqnarray}
\label{mainLT}
\mathcal{L}_{I_r}(s)&\stackrel{(a)}{=}&\mathbb{E}\Big[\prod\limits_{z\in \Psi_u\backslash \{x\}}\mathbb{E}_{R_z}\Big[\frac{1}{1+s \min\{P_{T_{\max}}, \overline{P}_T R_z^{\alpha} \} D_z^{-\alpha}}\Big]\Big]\nonumber\\
&\stackrel{(b)}{\approx}& \exp{\Big(-\int_{y>0}{ \! \Big(1-\mathbb{E}_{R_z}\Big[\Big.\frac{1}{1+s \min\{P_{T_{\max}}, \overline{P}_T R_z^{\alpha} \} y^{-\alpha}}\Big \vert R_z<y \Big]\Big) \,\Lambda_u(\mathrm{d}y)}\Big)}\nonumber\\
&\stackrel{(c)}{=}& \exp{\Big(-\int_{y>0}{ \! p\mathbb{E}_{R_z}\Big[\Big.\frac{1}{1+s^{-1} \overline{P}_T^{-1} R_z^{-\alpha}  y^{\alpha}}\Big \vert R_z<y, R_z<r_c \Big] \,\Lambda_u(\mathrm{d}y)}\Big)}\nonumber\\
&&\exp{\Big(-\int_{y>0}{ \! (1-p)\mathbb{E}_{R_z}\Big[\Big.\frac{1}{1+s^{-1} P_{T_{\max}}^{-1} y^{\alpha}}\Big \vert R_z<y, R_z>r_c \Big] \,\Lambda_u(\mathrm{d}y)}\Big)}\nonumber\\
&\stackrel{(d)}{=}& \exp{\Big(-\pi\lambda_a p\mathbb{E}_{R_z}\Big[\Big.R_z^2\int_{1}^{\infty}{ \! \frac{1}{1+s^{-1} \overline{P}_T^{-1} t^{\alpha/2}}  \,\mathrm{d}t}\Big \vert R_z<r_c\Big]\Big)}\nonumber\\
&&\exp{\Big(-\pi\lambda_a(1-p)\mathbb{E}_{R_z}\Big[\Big.\int_{1}^{\infty}{ \! R_z^2\frac{1}{1+s^{-1} P_{T_{\max}}^{-1} t^{\alpha/2}R_z^{\alpha}}   \,\mathrm{d}t}\Big \vert R_z>r_c\Big]\Big)}\nonumber\\
&\stackrel{(e)}{=}& \exp{\Big(-\pi\lambda_a p\mathbb{E}_{R_z}\left[R_z^2\vert R_z<r_c\right]\frac{2}{\alpha-2}s\overline{P}_T C_{\alpha}(s\overline{P}_T)\Big)}\nonumber\\
&&\exp{\Big(-\pi\lambda_a(1-p)\mathbb{E}_{R_z}\Big[\Big.R_z^2\frac{2}{\alpha-2}sP_{T_{\max}}R_z^{-\alpha}C_{\alpha}(sP_{T_{\max}}R_z^{-\alpha})\Big\vert R_z>r_c\Big]\Big)}\nonumber\\
&\stackrel{(f)}{=}& \exp\Big(- \frac{2s}{\alpha-2}\Big((1-e^{-\pi\lambda_ar_c^2}(1+\pi\lambda_ar_c^2))\overline{P}_TC_{\alpha}(s\overline{P}_T)\nonumber\\
&+&(1-p)\pi\lambda_aP_{T_{\max}}\mathbb{E}_{R_z}\left[\left.R_z^{2-\alpha}C_{\alpha}(sP_{T_{\max}}R_z^{-\alpha})\right\vert R_z>r_c\right]\Big)\Big),
\end{eqnarray}
where ($a$) follows from the iid nature of $\{g_z\}$ and the independence of $\{R_z\}$ (see Assumption \ref{distanceassumption}). The process $\Psi_u$ is not a PPP but a Poisson-Voronoi perturbed lattice and
hence the functional form of the interference (or the Laplace
functional of $\Psi_u$) is not tractable \cite{Singh2014}. Authors in \cite{Singh2014}
propose an approximation to characterize the corresponding
process as an inhomogeneous PPP with intensity measure
function $\Lambda_u(dy) = 2\pi\lambda_ay(dy)$. Hence, ($b$) follows from the definition of probability generating functional (PGFL) of the PPP \cite{Stoyan1996}, the independent path loss between the device and its serving aggregator \cite{Singh2014}, i.e., $R_z^{\alpha}$'s are independent, and the assumption that interferer's path loss is bounded as $R_z^{\alpha} < y^{\alpha}$ \cite{Singh2014}, ($c$) follows from the maximum power constraint, where $p=1-\exp\big(-\pi\lambda_a r_c^2\big)$ denotes the probability that $R_z<r_c$ and can be obtained using (\ref{Na}) and (\ref{Nad}), ($d$) follows from change of variables $t=(x/R_z)^2$, ($e$) follows using the definition of Gauss-Hypergeometric function, yielding $\int_{1}^{\infty}{ \! \frac{1}{1+s^{-1} t^{\alpha/2}}  \,\mathrm{d}t}=\frac{2s}{\alpha-2}{_2F_1}\Big(1,1-\frac{2}{\alpha},2-\frac{2}{\alpha},-s\Big)=\frac{2s}{\alpha-2}C_{\alpha}(s)$, ($f$) follows from Assumption \ref{distanceassumption} for $R_z$. Hence, $R_z^2$ is exponentially distributed with rate parameter $1/(2\sigma^2)=\pi\lambda_a$ that yields	
\begin{eqnarray} 
\label{RcondL}
p\mathbb{E}_{R_z}\left[R_z^2\vert R_z<r_c\right]=\int_{0}^{r_c^2}{\! v \pi\lambda_a e^{-\pi\lambda_a v}\, \mathrm{d}v}
=\frac{1-e^{-\pi\lambda_ar_c^2}(1+\pi\lambda_ar_c^2)}{\pi\lambda_a}.
\end{eqnarray}	

Conditioned on the distance between the device and its associated aggregator, 
\begin{eqnarray}
\mathbb{P}(\mathrm{SIR}>T\vert r<r_c)
=\mathbb{E}_{I_r}\left[\mathbb{P}\left(g >T\overline{P}_T^{-1}I_r\vert r<r_c\right)\right]\approx\mathcal{L}_{I_r}(T\overline{P}_T^{-1}),\nonumber\\
\mathbb{P}(\mathrm{SIR}>T\vert r>r_c)
=\mathbb{E}_{I_r}\left[\mathbb{P}\left(g >Tr^{\alpha}P_{T_{\max}}^{-1}I_r\vert r>r_c\right)\right]\approx\mathcal{L}_{I_r}(Tr^{\alpha}P_{T_{\max}}^{-1}).\nonumber
\end{eqnarray}
Hence, the uplink SIR coverage is obtained as $\mathcal{P}(T)
\approx p\mathcal{L}_{I_r}(T\overline{P}_T^{-1})+\int_{r_c}^{\infty}{\!\mathcal{L}_{I_r}(Tr^{\alpha}P_{T_{\max}}^{-1})f_R(r)\, \mathrm{d}r}$, where $R$ is Rayleigh distributed with parameter $\sigma=\sqrt{1/(2\pi\lambda_a)}$ from Assumption \ref{distanceassumption}.

\vspace{-0.2cm}
\subsection{Proof of Lemma \ref{parallelinterferencelemma}} \label{App:AppendixE}
Part (a). The intra-stage interference can be found using the density of active receiving devices of stage $k$, which is given by $\lambda_a^{\rm eff}(k)=p_{\mathrm{th}}(k)\lambda_a(k)$. The Laplace transform of $I_k$ is 
\begin{eqnarray}
\mathcal{L}_{I_k}(s)&=&\mathbb{E}\Big[\exp\Big(-s\sum\nolimits_{z_k\in \Psi_{u,k}\backslash \{x\}}g_z \min\{P_{T_{\max}}, \overline{P}_T R_{z_{k}}^{\alpha} \} D_{z_k}^{-\alpha}\Big)\Big]\nonumber\\
&{\approx}& 
\exp\Big(- \frac{2s}{\alpha-2}\Big((1-e^{-\pi\lambda_a^{\rm eff}(k)r_c^2}(1+\pi\lambda_a^{\rm eff}(k)r_c^2))\overline{P}_TC_{\alpha}(s\overline{P}_T)\nonumber\\
&+&(1-p_k)\pi\lambda_a^{\rm eff}(k)P_{T_{\max}}\mathbb{E}_{R_{z_k}}\left[\left.R_{z_k}^{2-\alpha}C_{\alpha}(sP_{T_{\max}}R_{z_k}^{-\alpha})\right\vert R_{z_k}>r_c\right]\Big)\Big).
\end{eqnarray}
This result is similar to the Laplace transform of the interference in (\ref{laplacetransformgeneral}), where the distribution of $R_{z_k}$ is Rayleigh, but with parameter $\sigma_k=\sqrt{1/(2\pi\lambda_a(k))}$ and $p_k=1-\exp(-\pi\lambda_a(k)r_c^2)$ is the probability that $R_{z,k}<r_c$ that decreases with increasing $k$.

Part (b). Let $I_{k^c}$ be the total inter-stage interference at stage $k$, and $\Psi_{u,k^c}=\cup_{j\neq k}\Psi_{u(j)}$ be the point process denoting the location of inter-stage devices transmitting on the same resource as the typical device. Note that the transmitter processes for all stages are determined by independent thinning of the initial device process $\Psi$ modeled as PPP. Hence, $\Psi_{u,k^c}$ is the superposition of the inter-stage transmitter processes, and is a PPP with density $\lambda_u^{\rm eff}(k^c)=\sum_{j\neq k}p_{\mathrm{th}}(j)\lambda_u(j)$, and is also independent of $\Psi_{u(k)}$ since the random variables $\{R_{z_k}\}$ are assumed independent \cite{Novlan2013}. Then, the Laplace transform of $I_{k^c}$, i.e., $\mathcal{L}_{I_{k^c}}(s)$, is given by
\begin{eqnarray}
\label{laplaceinterstage}
\mathcal{L}_{I_{k^c}}(s)
&=&\mathbb{E}\Big[\exp\Big(-s\sum\limits_{l\in k^c}{\sum\limits_{z_l\in \Psi_{u,l}}g_{z_l} \min\{P_{T_{\max}}, \overline{P}_T R_{z_l}^{\alpha} \} D_{z_l}^{-\alpha}}\Big)\Big]\nonumber\\
&\stackrel{(a)}{=}&\prod\limits_{l\in k^c}{\mathbb{E}\Big[\exp\Big(-s\sum\limits_{z_l\in \Psi_{u,l}}g_{z_l} \min\{P_{T_{\max}}, \overline{P}_T R_{z_l}^{\alpha} \} D_{z_l}^{-\alpha}\Big)\Big]}\nonumber\\
&{\approx}& \prod\limits_{l\in k^c}{\exp\Big(-\pi\lambda_a^{\rm eff}(l) p_l\mathbb{E}_{R_{z_l}}\Big[\Big.R_{z_l}^2\int_{0}^{\infty}{ \! \frac{1}{1+s^{-1} \overline{P}_T^{-1} t^{\alpha/2}}  \,\mathrm{d}t}\Big \vert R_{z_l}<r_c\Big]\Big)}\nonumber\\
&&\prod\limits_{l\in k^c}{\exp\Big(-\pi\lambda_a^{\rm eff}(l)(1-p_l)\mathbb{E}_{R_{z_l}}\Big[\Big.\int_{0}^{\infty}{ \! R_{z_l}^2\frac{1}{1+s^{-1} P_{T_{\max}}^{-1} t^{\alpha/2}R_{z_l}^{\alpha}}   \,\mathrm{d}t}\Big \vert R_{z_l}>r_c\Big]\Big)}\nonumber\\
&{=}& \prod\limits_{l\in k^c}\exp\Big(- (1-e^{-\pi\lambda_a^{\rm eff}(l)r_c^2}(1+\pi\lambda_a^{\rm eff}(l)r_c^2))(B_{\alpha}(s \overline{P}_T)+\frac{2s \overline{P}_T}{\alpha-2}C_{\alpha}(s \overline{P}_T))
\\
&&
-(1-p_l)\pi\lambda_a^{\rm eff}(l)\mathbb{E}_{R_{z_l}}\Big[R_{z_l}^2\left.\Big(B_{\alpha}\Big(\frac{s P_{T_{\max}}}{R_{z_l}^{\alpha}}\Big)+\frac{2sP_{T_{\max}}}{(\alpha-2)R_{z_l}^{\alpha}}C_{\alpha}\Big(\frac{s P_{T_{\max}}}{R_{z_l}^{\alpha}}\Big)\Big)\right \vert R_{z_l}>r_c\Big]\Big),\nonumber
\end{eqnarray}
where $R_{z_l}$ is Rayleigh with parameter $\sigma_l=\sqrt{1/(2\pi\lambda_a(l))}$, $p_l=1-\exp(-\pi\lambda_a(l)r_c^2)$ is the probability that $R_{z,l}<r_c$, and ($a$) follows from the independence of $\{R_{z_l}\}$ over stages $\{l\}$.

\end{appendix}

\begin{spacing}{1.35}
\bibliographystyle{IEEEtran}
\bibliography{M2Mreferences}  

\begin{thebibliography}{10}
\providecommand{\url}[1]{#1}
\csname url@samestyle\endcsname
\providecommand{\newblock}{\relax}
\providecommand{\bibinfo}[2]{#2}
\providecommand{\BIBentrySTDinterwordspacing}{\spaceskip=0pt\relax}
\providecommand{\BIBentryALTinterwordstretchfactor}{4}
\providecommand{\BIBentryALTinterwordspacing}{\spaceskip=\fontdimen2\font plus
\BIBentryALTinterwordstretchfactor\fontdimen3\font minus
  \fontdimen4\font\relax}
\providecommand{\BIBforeignlanguage}[2]{{%
\expandafter\ifx\csname l@#1\endcsname\relax
\typeout{** WARNING: IEEEtran.bst: No hyphenation pattern has been}%
\typeout{** loaded for the language `#1'. Using the pattern for}%
\typeout{** the default language instead.}%
\else
\language=\csname l@#1\endcsname
\fi
#2}}
\providecommand{\BIBdecl}{\relax}
\BIBdecl

\bibitem{Cisco2015}
``Cisco visual networking index: {Global} mobile data traffic forecast update,
  2014-2019,'' white paper, Cisco, Feb. 2015.

\bibitem{DhiHuVis2015}
H.~S. Dhillon, H.~Huang, and H.~Viswanathan, ``Wide-area wireless communication
  challenges for the {Internet of Things},'' \emph{submitted, arXiv preprint
  arXiv:1504.03242}, 2015.

\bibitem{Lien2011}
S.-Y. Lien, K.-C. Chen, and Y.~Lin, ``Toward ubiquitous massive accesses in
  3{GPP} machine-to-machine communications,'' \emph{IEEE Commun. Mag.},
  vol.~49, no.~4, pp. 66--74, Apr. 2011.

\bibitem{3GPPStandardv13}
\BIBentryALTinterwordspacing
\emph{New WI proposal: Further LTE Physical Layer Enhancements for MTC}, 3GPP
  Std. TDocs-RP-141\,660. [Online]. Available:
  \url{http://www.3gpp.org/DynaReport/TDocExMtg--RP-65--30566.htm}
\BIBentrySTDinterwordspacing

\bibitem{DhillonEH2-2013}
H.~S. Dhillon, H.~C. Huang, H.~Viswanathan, and R.~A. Valenzuela,
  ``Power-efficient system design for cellular-based machine-to-machine
  communications,'' \emph{IEEE Trans. Wireless Commun.}, vol.~12, no.~11, pp.
  5740--5753, Nov. 2013.

\bibitem{DhillonEH2014}
------, ``Fundamentals of throughput maximization with random arrivals for
  {M2M} communications,'' \emph{IEEE Trans. Commun.}, vol.~62, no.~11, pp.
  4094--4109, Nov. 2014.

\bibitem{Ester1996}
M.~Ester, H.~P. Kriegel, J.~Sander, and X.~Xu, ``A density-based algorithm for
  discovering clusters in large spatial databases with noise,'' in \emph{Proc.,
  KDD}, vol.~96, no.~34, Aug. 1996, pp. 226--231.

\bibitem{Bandyopadhyay2003}
S.~Bandyopadhyay and E.~J. Coyle, ``An energy efficient hierarchical clustering
  algorithm for wireless sensor networks.'' in \emph{Proc., IEEE Infocom},
  vol.~3, Apr. 2003, pp. 1713--1723.

\bibitem{Heinzelman2000energy}
W.~R. Heinzelman, A.~Chandrakasan, and H.~Balakrishnan, ``Energy-efficient
  communication protocol for wireless microsensor networks.'' in \emph{Proc.,
  Hawaii int. conf. on system sciences}, Jan. 2000.

\bibitem{Kaj2009}
I.~Kaj, ``Probabilistic analysis of hierarchical cluster protocols for wireless
  sensor networks,'' in \emph{Network Control and Optimization}.\hskip 1em plus
  0.5em minus 0.4em\relax Springer Berlin Heidelberg, 2009, pp. 137--151.

\bibitem{Younis2004}
O.~Younis and S.~Fahmy, ``{HEED}: A hybrid, energy-efficient, distributed
  clustering approach for ad hoc sensor networks,'' \emph{IEEE Trans. Mobile
  Comput.}, vol.~3, no.~4, pp. 366--379, Oct-Dec. 2004.

\bibitem{Qing2006}
L.~Qing, Q.~Zhu, and M.~Wang, ``Design of a distributed energy-efficient
  clustering algorithm for heterogeneous wireless sensor networks,''
  \emph{Computer Communications}, vol.~29, pp. 2230--2237, 2006.

\bibitem{Li2005MobileAdhoc}
C.~Li, M.~Ye, G.~Chen, and J.~Wu, ``An energy-efficient unequal clustering
  mechanism for wireless sensor networks,'' in \emph{Proc., IEEE Int. Conf. on
  Mobile Adhoc and Sensor Systems}, 2005.

\bibitem{Kwon2013}
T.~Kwon and J.~M. Cioffi, ``Random deployment of data collectors for serving
  randomly-located sensors,'' \emph{IEEE Trans. Wireless Commun.}, vol.~12,
  no.~6, pp. 2556--2565, Jun. 2013.

\bibitem{Cui2003}
S.~Cui, A.~J. Goldsmith, and A.~Bahai, ``Energy-constrained modulation
  optimization for coded systems.'' in \emph{Proc. IEEE Globecom}, vol.~1, Dec.
  2003, pp. 372--376.

\bibitem{Zaidi2012}
S.~A.~R. Zaidi, M.~Ghogho, D.~C. McLernon, and A.~Swami, ``Energy efficiency in
  large scale interference limited wireless ad hoc networks.'' in \emph{Proc.,
  IEEE Wireless Commun. and Networking Conf. Wksp (WCNCW)}, Apr. 2012, pp.
  24--29.

\bibitem{Zhou2008}
Z.~Zhou, S.~Zhou, S.~Cui, and J.-H. Cui, ``Energy-efficient cooperative
  communication in a clustered wireless sensor network,'' \emph{IEEE Trans.
  Veh. Technol.}, vol.~57, no.~6, Nov. 2008.

\bibitem{Holland2011}
M.~Holland, T.~Wang, B.~Tavli, A.~Seyedi, and W.~Heinzelman, ``Optimizing
  physical-layer parameters for wireless sensor networks.'' \emph{ACM Trans.
  Sensor Networks}, vol.~7, no.~4, 2011.

\bibitem{Hasan2013}
M.~Hasan, E.~Hossain, and D.~Niyato, ``Random access for machine-to-machine
  communication in {LTE}-advanced networks: issues and approaches,'' \emph{IEEE
  Commun. Mag.}, vol.~51, no.~6, pp. 86--93, 2013.

\bibitem{wu2013fasa}
H.~Wu, C.~Zhu, R.~J. La, X.~Liu, and Y.~Zhang, ``{FASA:} {Accelerated S-ALOHA}
  using access history for event-driven {M2M} communications,'' \emph{IEEE/ACM
  Trans. Netw}, vol.~21, no.~6, pp. 1904--1917, 2013.

\bibitem{Laya2014}
A.~Laya, L.~Alonso, and J.~Alonso-Zarate, ``Is the random access channel of
  {LTE} and {LTE-A} suitable for {M2M} communications? {A} survey of
  alternatives,'' \emph{IEEE Commun. Surveys Tuts.}, vol.~16, pp. 1--13, 2014.

\bibitem{Aijaz2014}
A.~Aijaz, M.~Tshangini, M.~Nakhai, X.~Chu, and H.~Aghvami, ``Energy-efficient
  uplink resource allocation in {LTE} networks with {M2M/H2H} co-existence
  under statistical {QoS} guarantees,'' \emph{IEEE Trans. Commun.}, vol.~62,
  pp. 2353--2365, 2014.

\bibitem{Rodoplu1999}
V.~Rodoplu and T.~Meng, ``Minimum energy mobile wireless networks,'' \emph{IEEE
  J. Sel. Areas Commun.}, vol.~17, 1999.

\bibitem{Baccelli1999}
F.~Baccelli and S.~Zuyev, ``Poisson-{Voronoi} spanning trees with applications
  to the optimization of communication networks,'' \emph{Operations Research},
  vol.~47, no.~4, pp. 619--631, 1999.

\bibitem{Baek2004}
S.~J. Baek, G.~de~Veciana, and X.~Su, ``Minimizing energy consumption in
  large-scale sensor networks through distributed data compression and
  hierarchical aggregation,'' \emph{IEEE J. Sel. Areas Commun.}, vol.~22,
  no.~6, pp. 1130--1140, Aug. 2004.

\bibitem{Haenggi2005}
M.~Haenggi and D.~Puccinelli, ``Routing in ad hoc networks: A case for long
  hops,'' \emph{IEEE Commun. Mag.}, vol.~43, 2005.

\bibitem{Stoyan1996}
D.~Stoyan, W.~Kendall, and J.~Mecke, \emph{Stochastic Geometry and Its
  Applications}, 2nd~ed.\hskip 1em plus 0.5em minus 0.4em\relax John Wiley and
  Sons, 1996.

\bibitem{Wang2006}
A.~Y. Wang and C.~Sodini, ``On the energy efficiency of wireless
  transceivers,'' in \emph{Proc., IEEE ICC}, 2006, pp. 3783--3788.

\bibitem{BaccelliBook1}
F.~Baccelli and B.~B{\l}aszczyszyn, \emph{Stochastic Geometry and Wireless
  Networks, Volume I --- Theory}, ser. Foundations and Trends in
  Networking.\hskip 1em plus 0.5em minus 0.4em\relax NoW Publishers, 2009,
  vol.~3, no. 3--4.

\bibitem{JaraiSzabo2008}
J.-S. Ferenc and Z.~Neda, ``On the size-distribution of {Poisson Voronoi}
  cells,'' \emph{Physica A: Statistical Mechanics and its Applications}, vol.
  385, no.~2, Nov. 2007.

\bibitem{Yu2013}
S.~M. Yu and S.-L. Kim, ``Downlink capacity and base station density in
  cellular networks,'' in \emph{in Proc., Int. Symp. Modeling \& Optimization
  in Mobile, Ad Hoc \& Wireless Networks (WiOpt)}, May 2013, pp. 119--124.

\bibitem{Novlan2013}
T.~D. Novlan, H.~S. Dhillon, and J.~G. Andrews, ``Analytical modeling of uplink
  cellular networks,'' \emph{IEEE Trans. Wireless Commun.}, vol.~12, no.~6, pp.
  2669--2679, Jun. 2013.

\bibitem{Singh2014}
S.~Singh, X.~Zhang, and J.~G. Andrews, ``Joint rate and {SINR} coverage
  analysis for decoupled uplink-downlink biased cell associations in
  {HetNets},'' \emph{IEEE Trans. Wireless Commun.}, vol.~14, no.~10, Oct. 2015.

\bibitem{Singh2013}
S.~Singh, H.~Dhillon, and J.~Andrews, ``Offloading in heterogeneous networks:
  {Modeling,} analysis, and design insights,'' \emph{IEEE Trans. Wireless
  Commun.}, vol.~12, no.~5, pp. 2484--2497, 2013.

\bibitem{Zhang2015}
X.~Zhang and J.~G. Andrews, ``Downlink cellular network analysis with
  multi-slope path loss models,'' \emph{IEEE Trans. Commun.}, vol.~63, no.~5,
  pp. 1881--1894, May 2015.

\end{thebibliography}
\end{spacing}

\end{document}